\newif\iffinal
\finalfalse	
\finaltrue 

\newif\ifmarek
\marektrue

\documentclass[reqno,twoside,11pt]{amsart}
\iffinal\else\usepackage[notref,notcite]{showkeys}\fi
\usepackage{cite}
\usepackage{lscape}
\usepackage[toc,title,page]{appendix}
\usepackage{subcaption}
\usepackage{amsmath}
\usepackage{amsmath,lipsum}
\usepackage{amsfonts}
\usepackage{amssymb}
\usepackage{bbm}
\usepackage{graphicx}
\usepackage{booktabs}
\usepackage{enumitem}
\usepackage{booktabs}
\usepackage{multirow}
\usepackage{tikz}
\usetikzlibrary{arrows,automata,positioning}
\graphicspath{{/Users/aubain2015/Documents/york/}}
\usepackage{verbatim}
\usepackage{color}
\usepackage{float}
\usepackage{dirtytalk}
\IfFileExists{epsf.def}{\input epsf.def}{\usepackage{epsf}}
\ifmarek
\IfFileExists{myowntimes.sty}{\usepackage{myowntimes}\usepackage{mathrsfs}}
	{\usepackage{times}\newcommand{\mathscr}{\mathcal}}

\usepackage{mathptmx}
\fi

\DeclareFontFamily{OT1}{eusb}{} \DeclareFontShape{OT1}{eusb}{m}{n}
{<5> <6> <7> <8> <9> <10> <11> <12> <14.4> eusb10}{}
\DeclareMathAlphabet{\eusb}{OT1}{eusb}{m}{n}

\DeclareFontFamily{OT1}{eusm}{} \DeclareFontShape{OT1}{eusm}{m}{n}
{<5> <6> <7> <8> <9> <10> <11> <12> <14.4> eusm10}{}
\DeclareMathAlphabet{\eusm}{OT1}{eusm}{m}{n}

\DeclareFontFamily{OT1}{eufm}{} \DeclareFontShape{OT1}{eufm}{m}{n}
{<5> <6> <7> <8> <9> <10> <11> <12> <14.4> eufm10}{}
\DeclareMathAlphabet{\mathfrak}{OT1}{eufm}{m}{n}

\DeclareFontFamily{OT1}{fraktura}{}
\DeclareFontShape{OT1}{fraktura}{m}{n} {<5> <6> <7> <8> <9> <10>
<11> <12> <13> <14.4> [1.1] eufm10}{}
\DeclareMathAlphabet{\fraktura}{OT1}{fraktura}{m}{n}

\DeclareFontFamily{OT1}{cmfi}{} \DeclareFontShape{OT1}{cmfi}{m}{n}
{<5> <6> <7> <8> <9> <10> <11> <12> <13> <14.4> [0.9] cmfi10}{}
\DeclareMathAlphabet{\cmfi}{OT1}{cmfi}{b}{n}

\DeclareFontFamily{OT1}{cmss}{} \DeclareFontShape{OT1}{cmss}{m}{n}
{<5> <6> <7> <8> <9> <10> <11> <12> <13> <14.4> cmss10}{}
\DeclareMathAlphabet{\cmss}{OT1}{cmss}{m}{n}

\newtheoremstyle{thm}{1.5ex}{1.5ex}{\itshape\rmfamily}{}
{\bfseries\rmfamily}{}{2ex}{}

\newtheoremstyle{def}{1.5ex}{1.5ex}{\slshape\rmfamily}{}
{\bfseries\rmfamily}{}{2ex}{}

\newtheoremstyle{rem}{1.3ex}{1.3ex}{\rmfamily}{}
{\itshape}
{} {1.5ex}{}

\theoremstyle{thm}
\newtheorem{theorem}{Theorem}[section]

\newtheorem{corollary}[theorem]{Corollary}

\theoremstyle{def}

\theoremstyle{rem}

\numberwithin{equation}{section}

\renewcommand{\subsection}{\secdef\subsct\sbsect}
\newcommand{\subsct}[2][default]{\refstepcounter{subsection}
\addcontentsline{toc}{subsection}
{{\tocsection{\!\!}{\hspace{1.2em}\thesubsection}{\!\!\!\!#1\dotfill}}{}}
\nopagebreak\vspace{0.45\baselineskip} {\flushleft\bf
\thesection.\arabic{subsection}~\bf #1.~}
\\*[3mm]\noindent
\nopagebreak}
\newcommand{\sbsect}[1]{\vspace{0.1cm}\noindent
\textbf{#1.~}\vspace{0.1cm}}

\renewcommand{\subsubsection}{%
\secdef \subsubsect\sbsbsect}
\newcommand{\subsubsect}[2][default]{%
\refstepcounter{subsubsection} 
\addcontentsline{toc}{subsubsection}{{\tocsection{\!\!}
{\hspace{3.05em}\thesubsubsection}{\!\!\!\!#1\dotfill}}{}}
\nopagebreak
\vspace{0.15\baselineskip} \nopagebreak {\flushleft\rmfamily
\itshape\arabic{section}.\arabic{subsection}.\arabic{subsubsection}
\ \rmfamily #1\/.}\ }
\newcommand{\sbsbsect}[1]{\vspace{0.1cm}\noindent
\rmfamily \itshape
\arabic{section}.\arabic{subsection}.\arabic{subsubsection} \
\sffamily #1\/.\ }
\iffinal

\else

\fi
   
\usepackage[toc,title,page]{appendix}

\newcommand{\N}{\mathbb N}

\newcommand{\R}{\mathbb R}

\newcommand{\scrF}{\mathscr{F}}

\title[]
{\large Fitting Generalized Tempered Stable distribution: \\
Fractional Fourier Transform (FRFT) Approach}

\begin{document}
\maketitle
\author{\centerline{A.H.Nzokem and V.T.Montshiwa}}

\begin{abstract}
The paper investigates the rich class of Generalized Tempered Stable (GTS) distribution, an alternative to Normal distribution and $\alpha$-Stable distribution for modelling asset return and many physical and economic systems. Firstly, we explore some important properties of the Generalized Tempered Stable  (GTS) distribution. The theoretical tools developed are used to perform empirical analysis. The GTS distribution is fitted using S\&P 500 index,  SPY ETF and Bitcoin BTC. The Fractional Fourier Transform (FRFT) technique evaluates the probability density function and its derivatives in the maximum likelihood procedure. Based on the results from the statistical inference and the Kolmogorov-Smirnov (K-S) goodness-of-fit, the GTS distribution fits the underlying distribution of the SPY ETF return. The right side of the Bitcoin BTC return, and the left side of the S\&P 500 return underlying distributions fit the Tempered Stable distribution; while the left side of the Bitcoin BTC return and the right side of the S\&P 500 return underlying distributions are modelled by the compound Poisson process.\\
  \noindent
\end{abstract}

 \section {Introduction}
\noindent
Modelling the return assets with the normal distribution is the underlying assumption in many financial tools, such as the Black-Scholes-Merton option pricing model and the risk metric Variance-covariance technique to Value at Risk (VAR). However, substantial empirical evidence rejects the normal distribution for various asset classes and financial markets. The normal distribution's symmetric and rapidly decreasing tail properties cannot describe the skewed and fat-tailed properties of the high-frequency asset return distribution.\\
\noindent
The $\alpha$- stable distribution has been proposed\cite{ken1999levy,nolan2020univariate} as an alternative to the normal distribution for modelling asset return and many types of physical and economic systems. Besides the Theoretical and empirical arguments for modelling with a stable distribution, the third and most important argument is that the Central Limit theorem can be generalized by the stable distribution. Regardless of the variance nature (finite or infinite), an appropriately standardized large sum of independent identical distribution (i.i.d) random variable converges to an $\alpha$- stable random variable\cite{rachev2011financial}. Although the stable distribution allows for varying degrees of tail heaviness and skewness; it has two major drawbacks \cite{nolan2020univariate}: firstly, the lack of closed formulas for densities and distribution functions, except for the Normal distribution($\alpha=2$), the Cauchy distribution ($\alpha=1$) and the L\'evy distribution ($\alpha=\frac{1}{2}$); secondly, most of the moments of the stable distribution are infinite. An infinite variance of the asset return leads to infinite price for derivative instruments such as options.\\
\noindent
The Tempered Stable (TS) distribution TS(\textbf{$\beta$},\textbf{$\alpha$},\textbf{$\lambda$}) was developed to overcome the shortcomings from those two distributions speciﬁcally in modelling high-frequency asset returns. The tails of the TS distribution for asset returns are heavier than the normal distribution but thinner than the stable distribution\cite{grabchak2010financial}. A more general form of the Tempered Stable distribution, called Generalized Tempered Stable (GTS) distribution GTS(\textbf{$\beta_{+}$}, \textbf{$\beta_{-}$}, \textbf{$\alpha_{+}$},\textbf{$\alpha_{-}$}, \textbf{$\lambda_{+}$}, \textbf{$\lambda_{-}$}) will be considerated in the paper. The L\'evy measure ($V(dx)$) of the Generalized Tempered Stable distribution is provided by (\ref{eq:l1}).
 \begin{align}
V(dx) =\left(\frac{\alpha_{+}e^{-\lambda_{+}x}}{x^{1+\beta_{+}}} \boldsymbol{1}_{x>0} + \frac{\alpha_{-}e^{-\lambda_{-}|x|}}{|x|^{1+\beta_{-}}} \boldsymbol{1}_{x<0}\right) dx \label{eq:l1}
 \end{align}
\noindent 
where $0\leq \beta_{+}\leq 1$, $0\leq \beta_{-}\leq 1$, $\alpha_{+}\geq 0$, $\alpha_{-}\geq 0$, $\lambda_{+}\geq 0$ and  $\lambda_{-}\geq 0$.\\
\noindent
The Generalized Tempered Stable Distribution can be used to control the level of skewness, tail heaviness and symmetry of the distribution. In addition, It has finite moments; just like the stable distributions, the closed-form formulas exist only for characteristic function and not for the density or the distribution function.\\
\noindent
The paper aims to evaluate the assumption that the asset return distribution follows a Generalized Tempered Stable distribution GTS(\textbf{$\beta_{+}$},\textbf{$\beta_{-}$},\textbf{$\alpha_{+}$},\textbf{$\alpha_{-}$},\textbf{$\lambda_{+}$},\textbf{$\lambda_{-}$}) and to compare their fitting performance to the Normal distribution , which is the standard in practice. The statistical inference is based on the Maximum Likelihood (ML) method. The Fractional Fourier Transform (FRFT) technique will be implemented to provide a good approximation of the density function and its derivatives in the optimization process. The data comes from three sources: the Standard $\&$ Poors 500 Composite Stock Price Index (S$\&$P 500), the SPDRS $\&$ P 500 ETF (SPY ETF) and the Bitcoin BTC. and the period spans from January 4, 2010, to February 02, 2022.
\noindent
The paper is structured as follows: the following section presents the theoretical framework of the Generalized Tempered Stable (GTS) distribution. The characteristic exponent closed-form formula is developed base on the L\'evy measure of the GTS distribution. The third section reviews the Maximum likelihood method implemented in the optimization process. The fourth section presents the Generalized Tempered Stable (GTS) distribution fitting results and the theoretical moments. And in the fifth section, the GTS distribution and the sample data distribution are compared through the Kolmogorov-Smirnov (KS) goodness-of-fit test.
\section{Tempered Stable Distribution}
\noindent
The L\'evy measure of the generalized tempered stable (GTS) distribution ($V(dx)$) is defined in (\ref{eq:l23}) as  a product of a tempering function ($q(x)$) and a L\'evy measure of the $\alpha$-stable distribution ($V_{stable}(dx)$).
 \begin{align}
q(x) &= e^{-\lambda_{+}x} \boldsymbol{1}_{x>0} + e^{-\lambda_{-}|x|} \boldsymbol{1}_{x<0} \label{eq:l21}\\
V_{stable}(dx) &=\left(\frac{\alpha_{+}}{x^{1+\beta_{+}}} \boldsymbol{1}_{x>0} + \frac{\alpha_{-}}{|x|^{1+\beta_{-}}} \boldsymbol{1}_{x<0}\right) dx \label{eq:l22}\\
V(dx) =q(x)V_{stable}(dx)&=\left(\frac{\alpha_{+}e^{-\lambda_{+}x}}{x^{1+\beta_{+}}} \boldsymbol{1}_{x>0} + \frac{\alpha_{-}e^{-\lambda_{-}|x|}}{|x|^{1+\beta_{-}}} \boldsymbol{1}_{x<0}\right) dx \label{eq:l23}
 \end{align}
\noindent 
where $0\leq \beta_{+}\leq 1$, $0\leq \beta_{-}\leq 1$, $\alpha_{+}\geq 0$, $\alpha_{-}\geq 0$, $\lambda_{+}\geq 0$ and  $\lambda_{-}\geq 0$. $\alpha_{+}$ and $\alpha_{-}$ are the scale parameter, also called the process intensity \cite{boyarchenko2002non} and has a similar role as the variance parameter in the Brownian motion process. These parameters play an essential role in the Levy process. $\lambda_{+}$ and $\lambda_{+}$ control the rate of decay on the positive and negative tails, respectively. When $\lambda_{+}>\lambda_{-}$, the distribution is skewed to the left; the left tail is thicker. When $\lambda_{+} < \lambda_{-}$, the distribution is skewed to the right. And when $\lambda_{+}=\lambda_{-}$, the distribution is symmetric \cite{rachev2011financial}. $\beta_{+}$ and $\beta_{-}$ are the indexes of stability, also called tail indexes, tail exponents or characteristic exponents, which determine the rate at which the tails of the distribution taper off \cite{borak2005stable}. See \cite{kuchler2013tempered,rachev2011financial} For more details on tempering function and L\'evy measure of tempered stable distribution.\\
\noindent
The Generalized Tempered Stable (GTS) distribution can be denoted by $X\sim GTS(\textbf{$\beta_{+}$}, \textbf{$\beta_{-}$}, \textbf{$\alpha_{+}$},\textbf{$\alpha_{-}$}, \textbf{$\lambda_{+}$}, \textbf{$\lambda_{-}$})$ where $X=X_{+} - X_{-}$ with $X_{+} \geq 0$, $X_{-}\geq 0$. $X_{+}\sim TS(\textbf{$\beta_{+}$}, \textbf{$\alpha_{+}$},\textbf{$\lambda_{+}$}, \textbf{$\lambda_{-}$})$ and  $X_{-}\sim TS(\textbf{$\beta_{-}$}, \textbf{$\alpha_{-}$},\textbf{$\lambda_{-}$}, \textbf{$\lambda_{-}$})$. 
\begin{align}
  \int_{-\infty}^{+\infty} V(dx) &=\begin{cases}
  +\infty   & \quad \text{if } {\beta_{+}\ge 0 \quad  \text{or} \quad \beta_{-} \ge 0}   \\
  \alpha_{+}{\lambda_{+}}^{\beta_{+}}\Gamma(-\beta_{+}) +  \alpha_{-}{\lambda_{-}}^{\beta_{-}}\Gamma(-\beta_{-})  & \quad \text{if }{ \beta_{+} < 0 \quad  \text{and } \quad \beta_{-} < 0 }\end{cases} \label{eq:l24}
     \end{align} 
From (\ref{eq:l24}), it results that when $\beta_{+} < 0$, TS(\textbf{$\beta_{+}$}, \textbf{$\alpha_{+}$},\textbf{$\lambda_{+}$}) is of finite activity and can be written as a Compound Poisson process on the right side ($X_{+}$). we have similar pattern when $\beta_{-} < 0$.\\
\noindent
However, when $0\le \beta_{+} \le 1$, $X_{+}$ is an infinite activities process with an infinite number of jumps in any given time interval. We have similar pattern when $0 \le \beta_{-} \le 1$. In addition to the infinite activities process, we have 
\begin{align}
  \int_{-\infty}^{+\infty} min(1,|x|)V(dx) <+ \infty \label{eq:l25} \end{align}
\noindent
By adding the location parameter of the distribution, the GTS distribution becomes GTS(\textbf{$\mu$}, \textbf{$\beta_{+}$}, \textbf{$\beta_{-}$}, \textbf{$\alpha_{+}$},\textbf{$\alpha_{-}$}, \textbf{$\lambda_{+}$}, \textbf{$\lambda_{-}$}) and we have:
 \begin{align}
Y=\mu +X \quad \quad  Y\sim GTS(\mu, \beta_{+}, \beta_{-}, \alpha_{+}, \alpha_{-},\lambda_{+}, \lambda_{-}) \label {eq:l26}
  \end{align}

\subsection{Tempered Stable Distribution and L\'evy Process}
\begin{theorem}\label{lem5} \ \\ 
Consider a variable $Y \sim GTS(\textbf{$\mu$}, \textbf{$\beta_{+}$}, \textbf{$\beta_{-}$}, \textbf{$\alpha_{+}$},\textbf{$\alpha_{-}$}, \textbf{$\lambda_{+}$}, \textbf{$\lambda_{-}$})$. The characteristic exponent can be written
  \begin{align}
\Psi(\xi)=\mu\xi i + \alpha_{+}\Gamma(-\beta_{+})\left((\lambda_{+} - i\xi)^{\beta_{+}} - {\lambda_{+}}^{\beta_{+}}\right) + \alpha_{-}\Gamma(-\beta_{-})\left((\lambda_{-} + i\xi)^{\beta_{-}} - {\lambda_{-}}^{\beta_{-}}\right) \label {eq:l27}
  \end{align}
\end{theorem} 
\begin{proof} \ \\
\noindent
(\ref{eq:l23}) is a L\'evy measure and (\ref{eq:l24}) is satisfied. We can applied the L\'evy-Khintchine representation for non-negative processes theorem. $Y=\mu +X =\mu + X_{+} - X_{-}$
\begin{align}
\Psi(\xi)&= Log\left(Ee^{i Y\xi}\right)=i\mu\xi + Log\left(Ee^{i X_{+}\xi}\right) + Log\left(Ee^{-i X_{-}\xi}\right) \label {eq:l28a}\\
&= i\mu\xi + \int_{0}^{+\infty} \left(e^{iy\xi} -1\right)\frac{\alpha_{+}e^{-\lambda_{+}y}}{y^{1+\beta_{+}}}dy + \int_{0}^{+\infty} \left(e^{-iy\xi} -1\right)\frac{\alpha_{-}e^{-\lambda_{-}y}}{y^{1+\beta_{-}}}dy \label{eq:l28b}\end{align}

 \begin{align}
 \int_{0}^{+\infty} \left(e^{iy\xi} -1\right)\frac{\alpha_{+}e^{-\lambda_{+}y}}{y^{1+\beta_{+}}}dy&=\alpha_{+}\lambda_{+}^{\beta_{+}}\Gamma(-\beta_{+})\sum_{k=1}^{+\infty}{\frac{\Gamma(k-\beta_{+})}{\Gamma(-\beta_{+})k!}(\frac{i\xi}{\lambda_{+}})^{k}}\label{eq:l28c}\\
 &=\alpha_{+}\lambda_{+}^{\beta_{+}}\Gamma(-\beta_{+})\sum_{k=1}^{+\infty}{\binom{\beta_{+}}{k}(-\frac{i\xi}{\lambda_{+}})^{k}}\label{eq:l28d}\\
 &=\alpha_{+}\Gamma(-\beta_{+}) \left((\lambda_{+} - i\xi)^{\beta_{+}} - \lambda_{+}^{\beta_{+}}\right)\label{eq:l28e} \end{align}
 Similarly, we have :
 \begin{align}
 \int_{0}^{+\infty} \left(e^{-iy\xi} -1\right)\frac{\alpha_{-}e^{-\lambda_{-}y}}{y^{1+\beta_{-}}}dy=\alpha_{-}\Gamma(-\beta_{-}) \left((\lambda_{-} + i\xi)^{\beta_{-}} - \lambda_{-}^{\beta_{-}}\right) \label{eq:l11g}
  \end{align} 
 \noindent 
 The expression in (\ref{eq:l28a}) becomes:
 \begin{align*}
\Psi(\xi)=i\mu\xi + \alpha_{+}\Gamma(-\beta_{+}) \left((\lambda_{+} - i\xi)^{\beta_{+}} - \lambda_{+}^{\beta_{+}}\right) + \alpha_{-}\Gamma(-\beta_{-}) \left((\lambda_{-} + i\xi)^{\beta_{-}} - \lambda_{-}^{\beta_{-}}\right)
\end{align*}
\end{proof}

\begin{theorem}\label{lem6} \ \\ 
Consider a variable $Y \sim GTS(\textbf{$\mu$}, \textbf{$\beta_{+}$}, \textbf{$\beta_{-}$}, \textbf{$\alpha_{+}$},\textbf{$\alpha_{-}$}, \textbf{$\lambda_{+}$}, \textbf{$\lambda_{-}$})$. \\
When $(\beta_{-},\beta_{+}) \to (0,0)$, GTS becomes Bilateral Gamma distribution with the characteristic exponent:
   \begin{align}
\Psi(\xi)=\mu\xi i - \alpha_{+}\log\left(1-\frac{1}{\lambda_{+}}i\xi\right)- \alpha_{-}\log\left(1+\frac{1}{\lambda_{+}}i\xi\right) \label {eq:l27a}
\end{align}
When $(\beta_{-},\beta_{+}) \to (0,0)$ and $\alpha_{-}=\alpha_{+}=\alpha$, GTS becomes Variance-Gamma (VG) distribution $(\mu,\lambda_{-}-\lambda_{+},1,\alpha,\frac{1}{\lambda_{-}\lambda_{+}})$ and the characteristic exponent:
   \begin{align}
\Psi(\xi)=\mu\xi i - \alpha \log \left(1 - \frac{\lambda_{-}-\lambda_{+}}{\lambda_{+}\lambda_{-}} i\xi +\frac{1}{\lambda_{+}\lambda_{-}}\xi^2 \right) \label {eq:l27b}
  \end{align}
\end{theorem} 
\begin{proof} 
\begin{align}
\Gamma(-\beta_{+})=-\frac{\Gamma(1-\beta_{+})}{\beta_{+}} \quad \quad \lim_{\beta_{+} \to 0} \Gamma(-\beta_{+})\left((\lambda_{+} - i\xi)^{\beta_{+}} - \lambda_{+}^{\beta_{+}}\right)=- \log\left(1-\frac{1}{\lambda_{+}}i\xi\right)\label {eq:l27c}\end{align}
Similarly, (\ref{eq:l27c}) works for  $\beta_{-} \to 0$. 
From (\ref{eq:l27}), the characteristic exponent becomes (\ref{eq:l27a}).\\
In addition, if $\alpha_{-}=\alpha_{+}=\alpha$, from (\ref{eq:l27a}), the characteristic exponent becomes (\ref{eq:l27b}), which is a Variance-Gamma distribution with parameter $(\mu,\lambda_{-}-\lambda_{+},1,\alpha,\frac{1}{\lambda_{-}\lambda_{+}})$. \\
See\cite{nzokem2022} for more details on Variance-Gamma (VG) model and VG parameter notation.
\end{proof}

 \begin{theorem}\label{lem7} (Cumulants $\kappa_{k}$)\\
 Consider a variable $Y \sim GTS(\textbf{$\mu$}, \textbf{$\beta_{+}$}, \textbf{$\beta_{-}$}, \textbf{$\alpha_{+}$},\textbf{$\alpha_{-}$}, \textbf{$\lambda_{+}$}, \textbf{$\lambda_{-}$})$. The cumulants $\kappa_{k}$ of the Generalized Tempered Stable distribution is defined as follows.
  \begin{align}
\kappa_{1}= \mu + \alpha_{+}{\frac{\Gamma(1-\beta_{+})}{\lambda_{+}^{1-\beta_{+}}}} - \alpha_{-}{\frac{\Gamma(1-\beta_{-})}{\lambda_{-}^{1-\beta_{-}}}} \quad
 \kappa_{k}=\alpha_{+}{\frac{\Gamma(k-\beta_{+})}{\lambda_{+}^{k-\beta_{+}}}} + (-1)^{k} \alpha_{-}{\frac{\Gamma(k-\beta_{-})}{\lambda_{-}^{k-\beta_{-}}}}  \label {eq:l12}
  \end{align}
 \end{theorem}
\begin{proof} \ \\
\noindent
we consider the characteristic exponent $\Psi(\xi)$  in (\ref{eq:l28b}).
 \begin{align}
\Psi(\xi)&= i\mu\xi  + \int_{0}^{+\infty} \left(e^{iy\xi} -1\right)\frac{\alpha_{+}e^{-\lambda_{+}y}}{y^{1+\beta_{+}}}dy + \int_{0}^{+\infty} \left(e^{-iy\xi} -1\right)\frac{\alpha_{-}e^{\lambda_{-}y}}{y^{1+\beta_{-}}}dy \label {eq:l12a}\\
&= i\mu\xi + \alpha_{+}\sum_{k=1}^{+\infty}{\frac{\Gamma(k-\beta_{+})}{\lambda_{+}^{k-\beta_{+}}}\frac{(i\xi)^{k}}{k!}} + \alpha_{-}\sum_{k=1}^{+\infty}{\frac{\Gamma(k-\beta_{-})}{\lambda_{-}^{k-\beta_{-}}}\frac{(-i\xi)^{k}}{k!}}\label {eq:l12b}\\
&= i\mu\xi + \sum_{k=1}^{+\infty}{\frac{1}{k!}\left(\alpha_{+}\frac{\Gamma(k-\beta_{+})}{\lambda_{+}^{k-\beta_{+}}} + \alpha_{-}\frac{\Gamma(k-\beta_{-})}{\lambda_{-}^{k-\beta_{-}}}(-1)^{k}\right)(i\xi)^{k}}=\sum_{k=0}^{+\infty}\frac{\kappa_{k}}{k!}(i\xi)^{k} \label {eq:l12d}
\end{align}
\noindent
Hence, the $k$-th order cumulant $\kappa_{k}$  is given by comparing the coefficients of both polynomial functions in (\ref{eq:l12d}) in  $i\xi$.
\end{proof}

\begin{corollary} \ \\
Let $Y=\left(Y_{t}\right)$ be a L\'evy process on $\mathbb{R^{+}}$ generated by $GTS(\textbf{$\mu$}, \textbf{$\beta_{+}$}, \textbf{$\beta_{-}$}, \textbf{$\alpha_{+}$},\textbf{$\alpha_{-}$}, \textbf{$\lambda_{+}$}, \textbf{$\lambda_{-}$})$. \\
$\forall t \in \R^{+}$, $Y_{t}$  follows a $GTS(t\textbf{$\mu$}, \textbf{$\beta_{+}$}, \textbf{$\beta_{-}$}, \textbf{$t\alpha_{+}$},\textbf{$t\alpha_{-}$}, \textbf{$\lambda_{+}$}, \textbf{$\lambda_{-}$})$
\end{corollary}

\begin{proof}[Proof:]  \ \\
\noindent
Let $\varphi(\xi,t)$ be the characteristic exponent  of the L\'evy process  $Y=\left(Y_{t}\right)$. We have 
 \begin{align*}
\Psi(\xi,t)&= Log\left(Ee^{i Y_{t}\xi}\right)= t Log\left(Ee^{i X\xi}\right)\\
&=t\mu\xi i + t\alpha_{+}\Gamma(-\beta_{+})\left( (\lambda_{+} - i\xi)^{\beta_{+}} - \lambda_{+}^{\beta_{+}}\right) + t\alpha_{-}\Gamma(-\beta_{-})\left( (\lambda_{-} + i\xi)^{\beta_{-}} - \lambda_{-}^{\beta_{-}}\right) 
   \end{align*}
We have:  $Y_{t} \sim GTS(t\textbf{$\mu$}, \textbf{$\beta_{+}$}, \textbf{$\beta_{-}$}, \textbf{$t\alpha_{+}$},\textbf{$t\alpha_{-}$}, \textbf{$\lambda_{+}$}, \textbf{$\lambda_{-}$})$

\end{proof} 
 \begin{theorem}\label{lem7} (Asymptotic distribution of Generalized Tempered Stable distribution  process)\\
Let $Y={Y_{t}}$ be a L\'evy process on $\mathbb{R}$ generated by  GTS(\textbf{$\mu$}, \textbf{$\beta_{+}$}, \textbf{$\beta_{-}$}, \textbf{$\alpha_{+}$},\textbf{$\alpha_{-}$}, \textbf{$\lambda_{+}$}, \textbf{$\lambda_{-}$}).
Then  $Y_{t}$  converges in distribution to a L\'evy process driving by a Normal distribution with mean $\kappa_{1}$ and variance $\kappa_{2}$
\begin{align*}
Y_{t}\overset{\text{d}}{\sim}N(t\kappa_{1},t\kappa_{2}) \quad &\text{as} \quad t \to +\infty
 \end{align*}
\end{theorem}

\begin{proof}[Proof:]  \ \\
\noindent
It was shown in\cite{kendall1946advanced} that $\Psi(\xi)$ generates the cumulants $\left(\kappa_{j}\right)_{j\in \N}$ in (\ref{eq:l12}) such that 
 \begin{align}
\Psi(\xi)=Log\left(Ee^{i Y\xi}\right)= \sum_{j=0}^{+\infty}\kappa_{j}\frac{(i\xi)^{j}}{j!}
 \end{align}

$\phi(\xi,t)$ is the characteristic function of the stochastic  process $\frac{{Y_{t}}-t\kappa_{1}}{\sqrt{t\kappa_{2}}}$ and we have 
 \begin{align*}
\phi^{T}(\xi,t)= E\{e^{i \frac{{Y_{t}}-t\kappa_{1}}{\sqrt{t\kappa_{2}}}\xi}\} &=e^{-i \frac{t\kappa_{1}}{\sqrt{t\kappa_{2}}}\xi}E\{e^{i \frac{\xi}{\sqrt{t\kappa_{2}}}Y_{t}}\}=e^{-i \frac{t\kappa_{1}}{\sqrt{t\kappa_{2}}}\xi}e^{t\Psi(\frac{\xi}{\sqrt{t\kappa_{2}}})}\\
&=e^{-i \frac{t\kappa_{1}}{\sqrt{t\kappa_{2}}}\xi}e^{\sum_{j=0}^{+\infty}\frac{t\kappa_{j}}{j!}(i\frac{\xi}{\sqrt{t\kappa_{2}}})^{j}}=e^{-\frac{\xi^{2}}{2} + \sum_{j=3}^{+\infty}\frac{t\kappa_{j}}{j!}(i\frac{\xi}{\sqrt{t\kappa_{2}}})^{j}}
   \end{align*}
When $t \to +\infty$
 \begin{align}
\lim_{t \to +\infty} \sum_{j=3}^{+\infty}\frac{t\kappa_{j}}{j!}(i\frac{\xi}{\sqrt{t\kappa_{2}}})^{j}=0 \quad \quad  
\lim_{t \to +\infty}\phi^{T}(\xi,t)=\lim_{t \to +\infty}e^{-\frac{\xi^{2}}{2} + \sum_{j=3}^{+\infty}\frac{t\kappa_{j}}{j!}(i\frac{\xi}{\sqrt{t\kappa_{2}}})^{j}} =e^{-\frac{1}{2}\xi^2} \label{eq:l555a}
\end{align}
\end{proof}
\section{Data and Methodology}
\subsection{Data Summaries}
 \noindent 
The Standard \& Poor’s 500 Composite Stock Price Index, also known as the S\&P 500, is a stock index that tracks the share prices of 500 of the largest public companies in the United States. It is often treated as a proxy for describing the overall health of the stock market or even the U.S. economy. The SPDR S\&P 500   ETF (SPY), also known as the SPY ETF,  is an Exchange-Traded Fund (ETF)that tracks the performance of the S\&P 500. Like S\&P, SPY provides the diversification of a mutual fund, the flexibility of a stock and lower trading fees.\\
The S\&P 500 and SPY ETF data were extracted from Yahoo finance. The daily prices were adjusted for splits and dividends. The period spans from January 4, 2010 to February 02, 2022. The Bitcoin BTC price was extracted from CoinMarketCap. The period spans from April 28, 2013 to April 08, 2022. The daily price dynamics are provided in Fig \ref{fig66}, Fig \ref{fig67} and Fig \ref{fig68} . Prices have an increasing trend, but the scale levels are different; S\&P 500 and Bitcoin are priced in thousand US dollars and SPY ETF in hundred US dollars. The increasing trend is temporally disrupted in the first quarter of 2020 by the coronavirus pandemic.\\
\\
\begin{figure}[ht]
\vspace{-0.5cm}
    \centering
  \begin{subfigure}[b]{0.3\linewidth}
    \includegraphics[width=\linewidth]{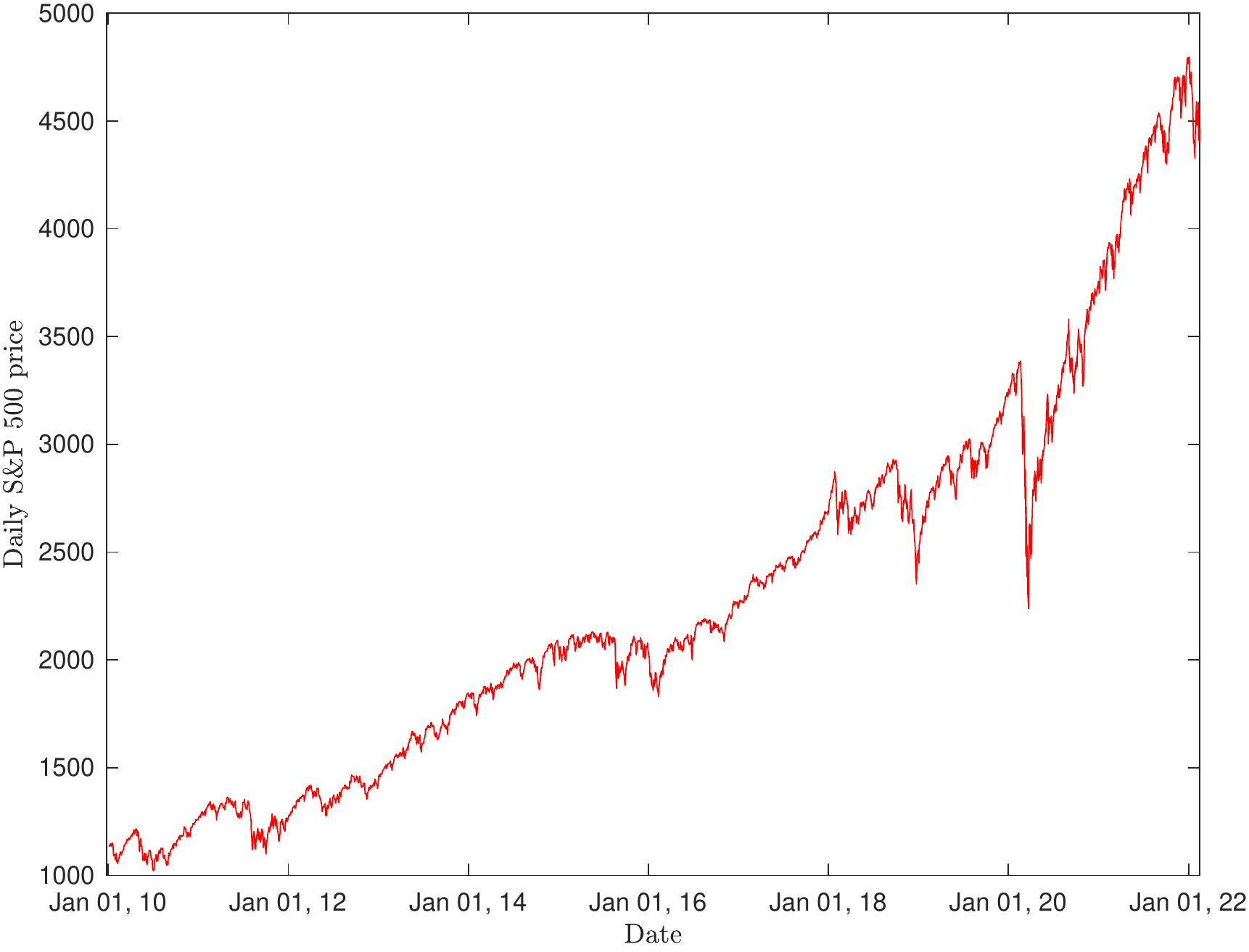}
\vspace{-0.5cm}
     \caption{Daily price of the S\&P 500}
         \label{fig66}
  \end{subfigure}
  \begin{subfigure}[b]{0.3\linewidth}
    \includegraphics[width=\linewidth]{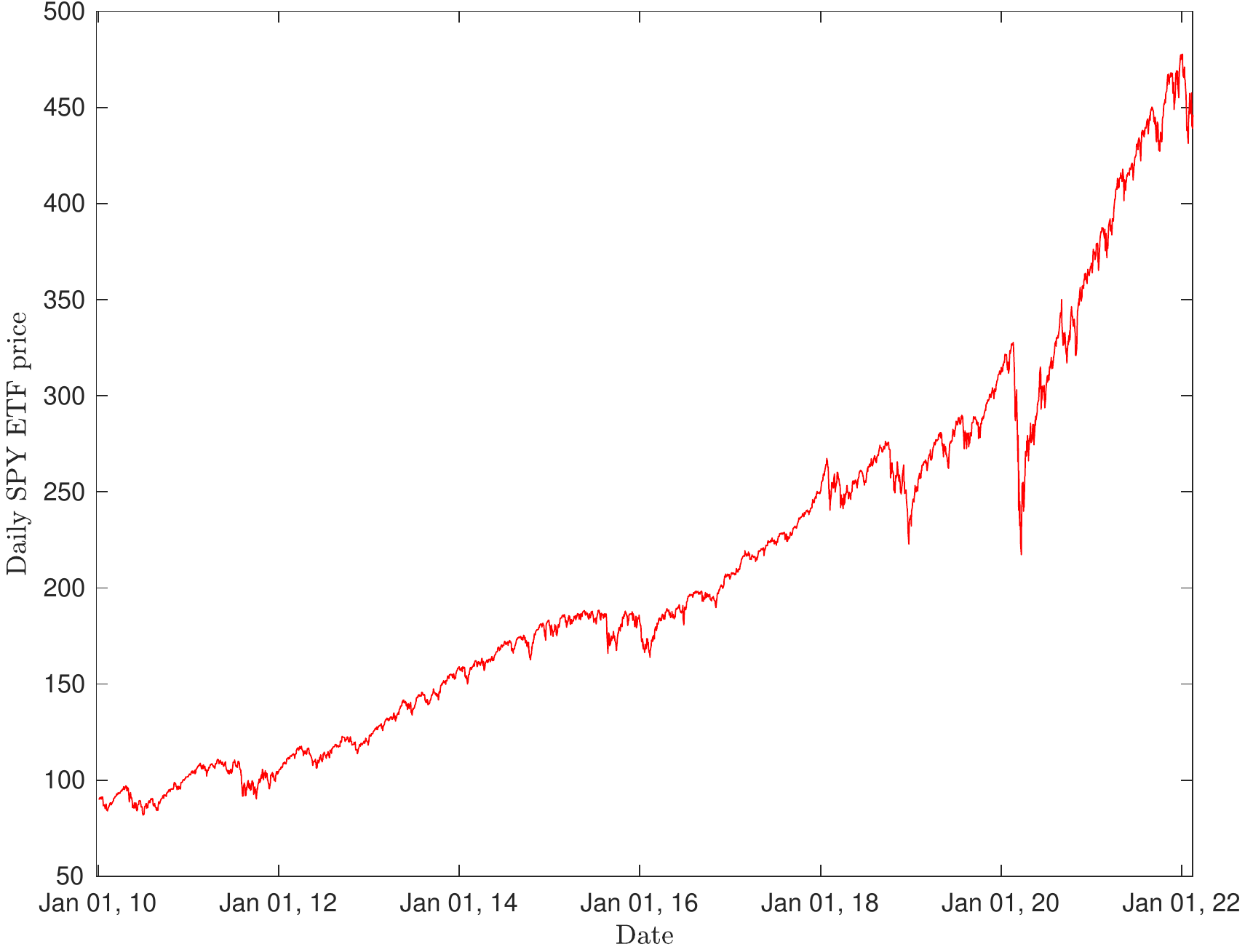}
\vspace{-0.5cm}
     \caption{Daily price of SPY ETF}
         \label{fig67}
          \end{subfigure}
  \begin{subfigure}[b]{0.3\linewidth}
    \includegraphics[width=\linewidth]{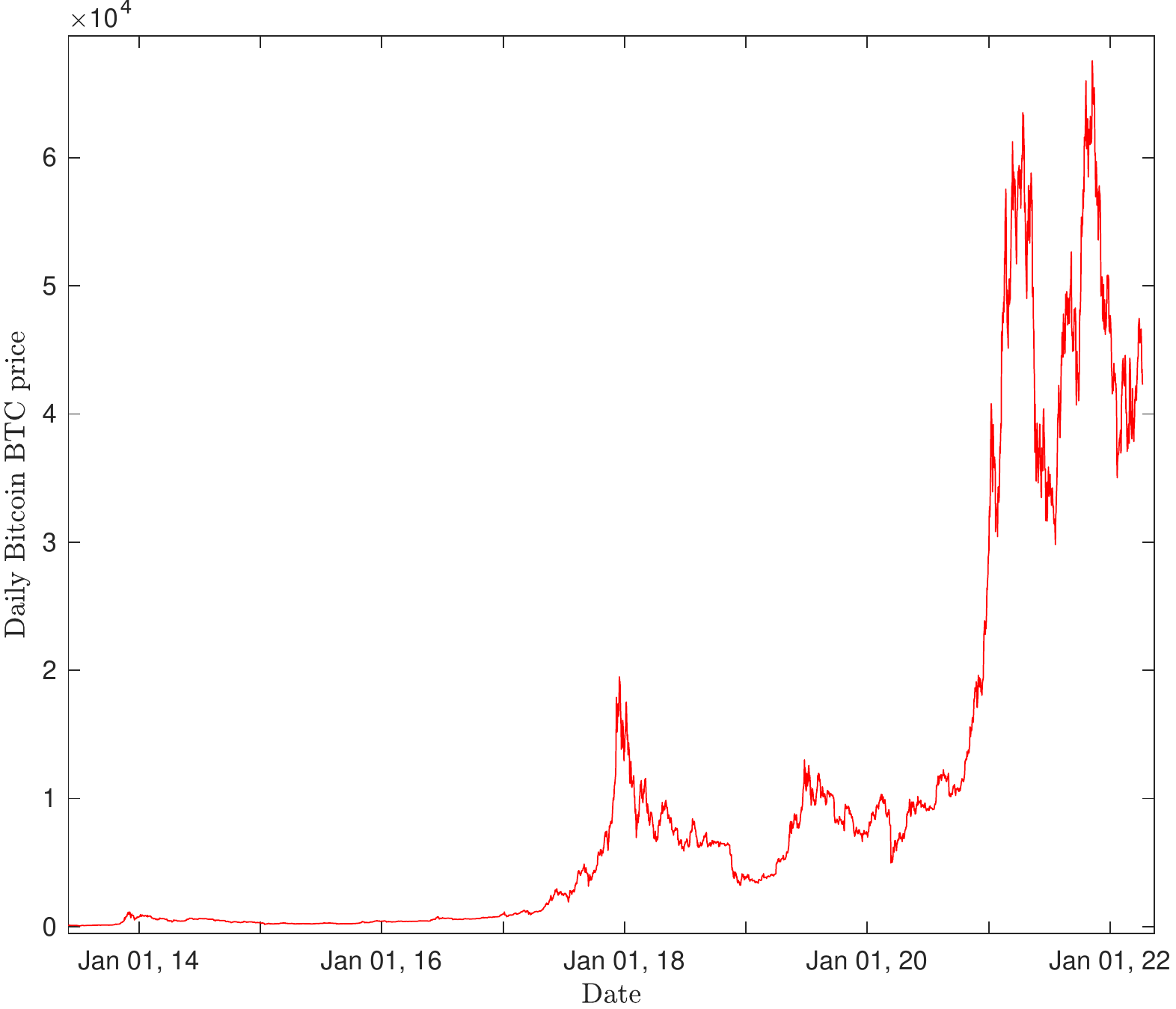}
\vspace{-0.5cm}
     \caption{Daily price of Bitcoin BTC}
         \label{fig68}
          \end{subfigure}
\vspace{-0.7cm}
  \caption{Daily Index Price}
  \label{fig66a}
\vspace{-0.5cm}
\end{figure}

\noindent
Let the number of observations $m$, and the daily observed price $S_{j}$ on day $t_{j}$ with $j=1,\dots,m$; $t_{1}$ is the first observation date (January 04, 2010) and $t_{m}$ is the last observation date (February 02, 2022). The daily return $(y_{j})$ is computed as in (\ref{eq:l31}).
\begin{align}
y_{j}=\log(S_{j}/S_{j-1}) \hspace{10 mm}  \hbox{ $j=2,\dots,m$}\label{eq:l31}
 \end{align}

\begin{figure}[ht]
\vspace{-0.5cm}
    \centering
  \begin{subfigure}[b]{0.31\linewidth}
    \includegraphics[width=\linewidth]{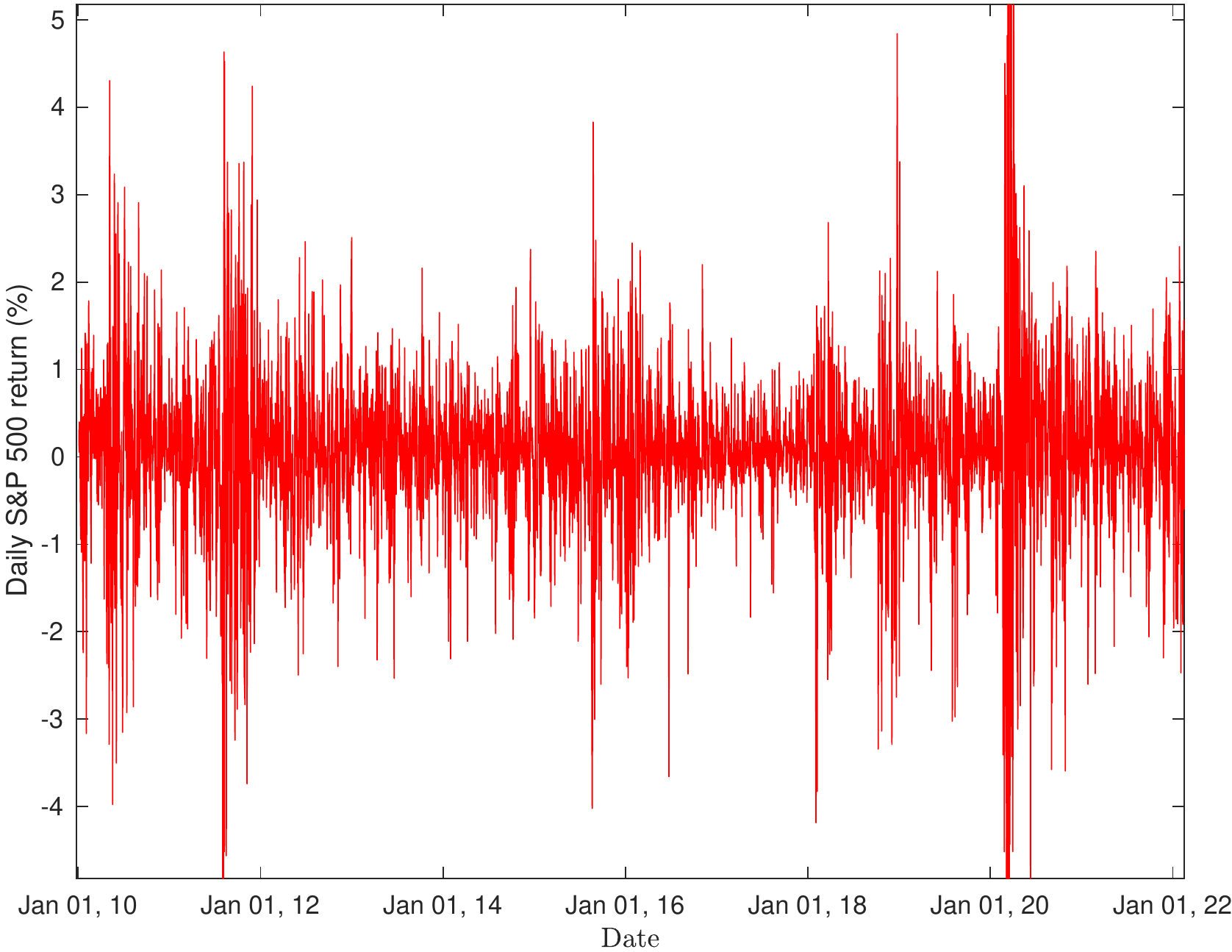}
\vspace{-0.5cm}
     \caption{Daily SP500 return}
         \label{fig69}
  \end{subfigure}
  \begin{subfigure}[b]{0.31\linewidth}
    \includegraphics[width=\linewidth]{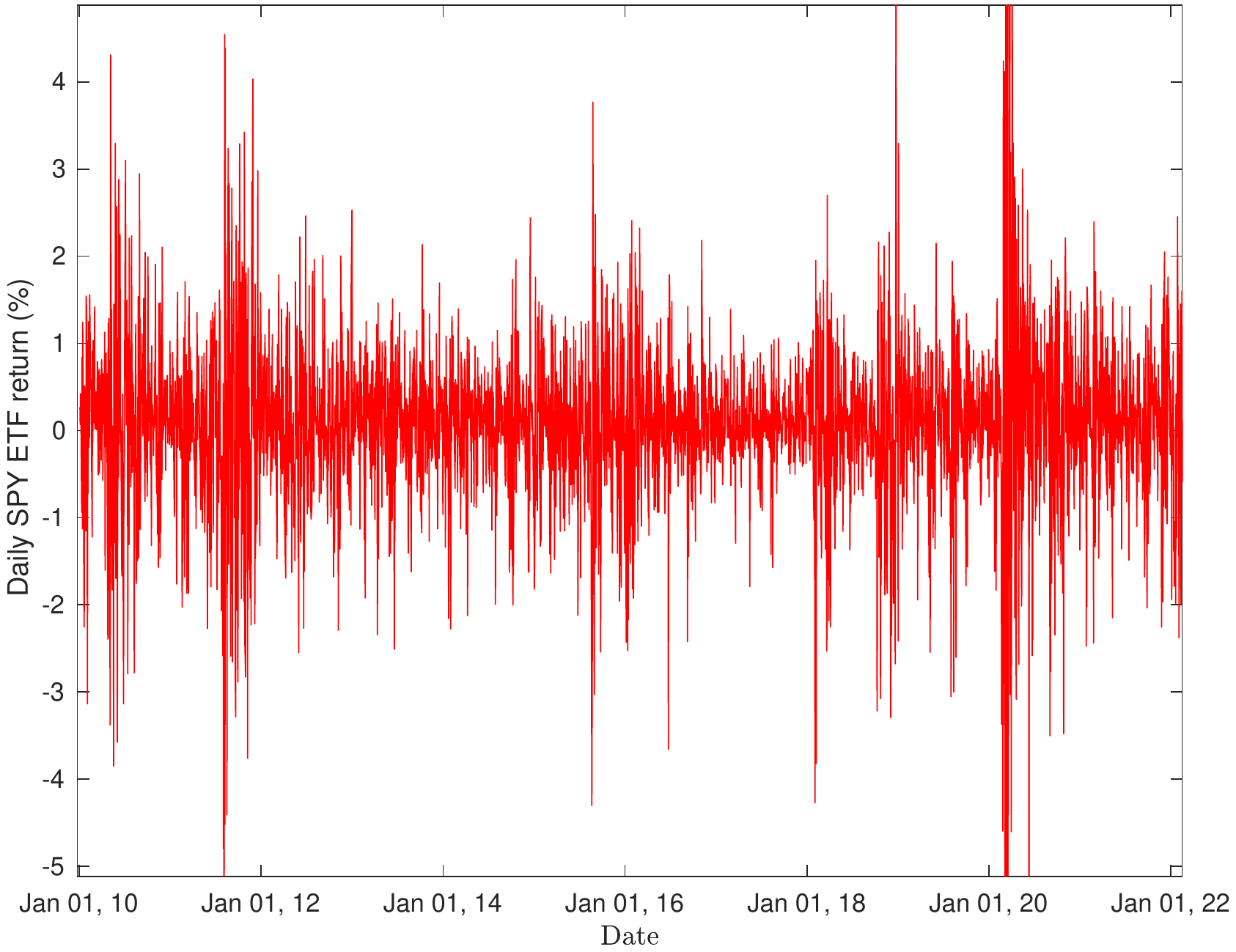}
\vspace{-0.5cm}
     \caption{ Daily SPY ETF return }
         \label{fig70}
          \end{subfigure}
 \begin{subfigure}[b]{0.31\linewidth}
    \includegraphics[width=\linewidth]{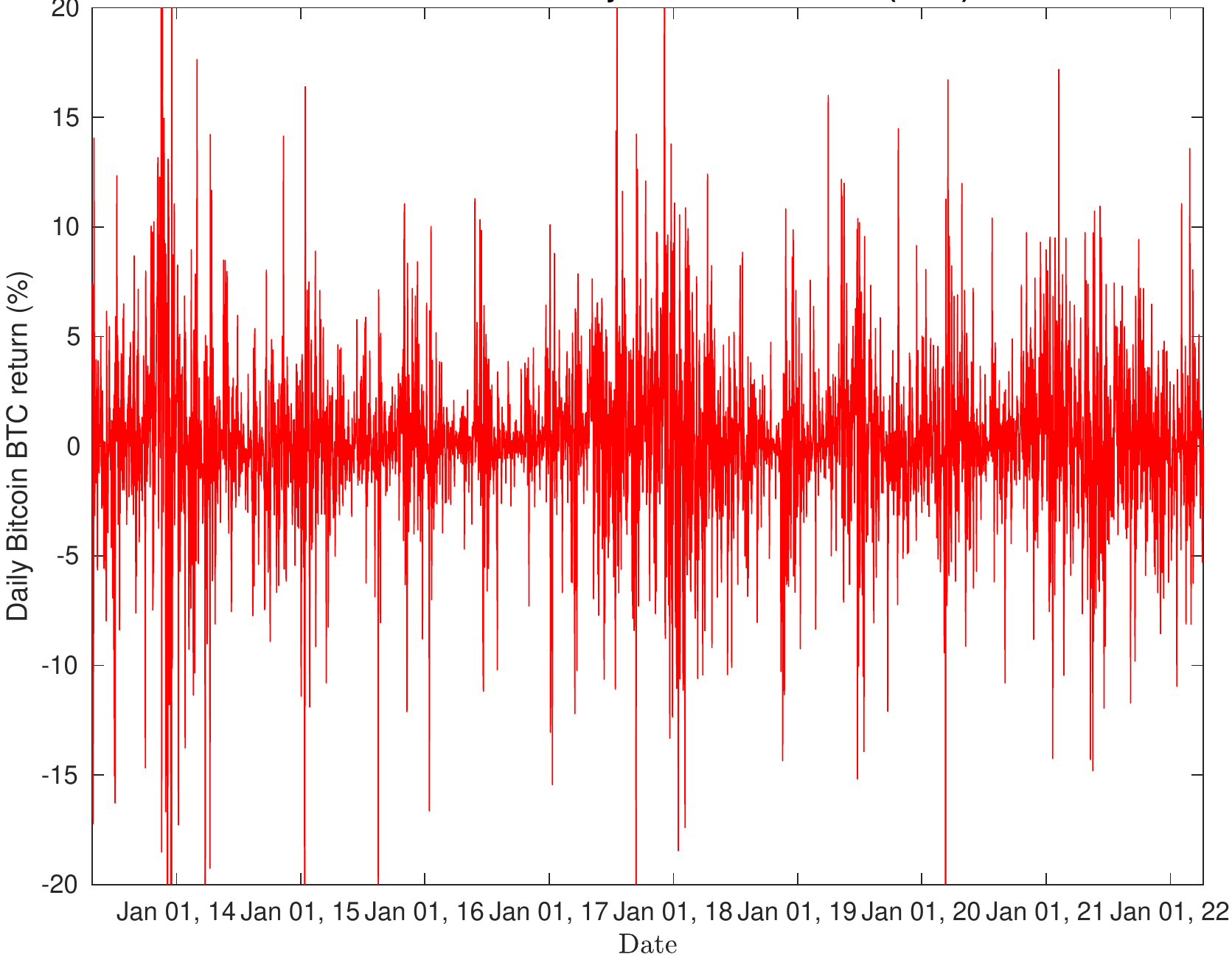}
\vspace{-0.5cm}
     \caption{ Daily Bitcoin BTC return }
         \label{fig71}
          \end{subfigure}
\vspace{-0.7cm}
  \caption{Daily SPY ETF Return ($\%$)}
  \label{fig66b}
\vspace{-0.5cm}
\end{figure}

\noindent  
As shown in Fig \ref{fig66b}, The volatility of each daily return is higher in the First quarter of $2020$ amid the coronavirus pandemic and massive disruptions in the global economy. SPY ETF aims to mirror the performance of the S\&P 500. Fig \ref{fig69} and Fig \ref{fig70} look similar, which is consistent with the goal of SPY ETF.\\
\noindent 
Some daily return observations were identified as outliers and removed from the data set to avoid a negative impact on statistics.  By assuming the independent, identically distribution of each return in the sample, the method of moments was used to compute statistics based on data without outliers. The data were summarized and provided statistical information in Table \ref{tab:1}.
\begin{table}[ht]
\vspace{-0.5cm}
 \caption{Summary Statistics  }
\label{tab:1}
\centering
\begin{tabular}{@{}l c c c c c@{}}
\toprule
 & \multicolumn{3}{c}{Empirical Statistics} \\\toprule
Label & S\&P 500 return & SPY ETF  return & Bitcoin return\\ \toprule
Sample size (m) & 3046 & 3046 & 3264\\
Mean  & 0.0501531 & 0.0546935 & 0.1790993\\
Variance & 0.9465432 & 0.9349634 & 16.2308643 \\
Skewness & -0.3627460 & -0.4587793 & -0.3368540  \\
Kurtosis &  6.4474002 & 6.4670659 &  7.9194813\\ \bottomrule
\end{tabular}%
\end{table}

\subsection{Review of the Maximum Likelihood Method}
\noindent
From a probability density function $f(x,V)$ with parameter $V=(\textbf{$\mu$}, \textbf{$\beta_{+}$}, \textbf{$\beta_{-}$}, \textbf{$\alpha_{+}$},\textbf{$\alpha_{-}$}, \textbf{$\lambda_{+}$}, \textbf{$\lambda_{-}$})$ of size $p=7$ and the sample data $x$ of size $m$,  we definite  the Likelihood Function and its derivatives as follows:  \begin{align}
 L(x,V) &= \prod_{j=1}^{m} f(x_{j},V) \quad & \quad 
 l(x,V) &= \sum_{j=1}^{m} log(f(x_{j},V))  \label {eq:l32} \\
 \frac{dl(x,V)}{dV_j} &= \sum_{i=1}^{m} \frac{\frac{df(x_{i},V)}{dV_j}}{f(x_{i},V)} \quad & \quad
  \frac{d^{2}l(x,V)}{dV_{k}dV_{j}} &= \sum_{i=1}^{m} \left(\frac{\frac{d^{2}f(x_{i},V)}{dV_{k}dV_{j}}}{f(x_{i},V)}- \frac{\frac{df(x_{i},V)}{dV_{k}}}{f(x_{i},V)}\frac{\frac{df(x_{i},V)}{dV_j}}{f(x_{i},V)}\right) \label {eq:l35}
 \end{align}
\noindent
In order to perform the Maximum of the likelihood function  (\ref{eq:l32}), the quantities  $\frac{dl(x,V)}{dV_j}$ and $\frac{d^{2}l(x,V)}{dV_{k}dV_{j}}$, which are the first and second order derivative of the density function respectively, must be computed. 
\noindent
The quantities $\frac{d^{2}l(x,V)}{dV_{k}dV_{j}}$ in (\ref{eq:l35}) are critical in computing the Hessian Matrix and the Fisher Information Matrix.\\
\noindent
Given the parameters $V=(\textbf{$\mu$}, \textbf{$\beta_{+}$}, \textbf{$\beta_{-}$}, \textbf{$\alpha_{+}$},\textbf{$\alpha_{-}$}, \textbf{$\lambda_{+}$}, \textbf{$\lambda_{-}$})$ and the sample data set $X$, we have the following quantities (\ref{eq:l36}) from the previous development and computations
  \begin{align}
 I'(X,V) =\left(\frac{dl(x,V)}{dV_j}\right)_{1 \leq j \leq p}   \quad  \quad  I''(X,V) = \left(\frac{d^{2}l(x,V)}{dV_{k}dV_{j}}\right)_{\substack{{1 \leq k \leq p} \\ {1 \leq j \leq p}}} \label {eq:l36}
 \end{align}
\noindent
 The Fractional Fourier Transform (FRFT) technique computes the likelihood function (\ref{eq:l32}) and its derivatives (\ref{eq:l36}) in the optimization process. The numerical integration technique, called the 12-point rule Composite Newton-Cotes Quadrature Formulas\cite{Nzokem_2021} can also be used. See \cite{Nzokem_2022,nzokem2021fitting} for a short review of the FRFT technique and the choice of FRFT parameters.\\
\noindent
we should have a local solution $v$ and a negative semi-definite matrix ($I''(x,V)$) when
  \begin{align}
 I'(x,V)=0 \quad \quad U^{T}\mathbf{I''(X,V)}U \leq 0\hspace{5mm}  \hbox{ $\forall U \in \R^{p}$}\label{eq:l38}
 \end{align}
The solutions of (\ref{eq:l38}) is provided by the Newton--Raphson Iteration Algorithm (\ref{eq:l39}). 
  \begin{align}
V^{n+1}=V^{n}-{\left(I''(x,V^{n})\right)^{-1}}I'(x,V^{n})\label{eq:l39}
 \end{align}
\noindent
For more detail on Maximum likelihood and Newton–Raphson Iteration procedure, see \cite{giudici2013wiley}.


\section{Fitting Tempered Stable distribution results}
\noindent 
 The Fractional Fourier Transform (FRFT) technique computes the likelihood function and its derivatives in the optimization process for a given asset return data. See \cite{Nzokem_2022,nzokem2021fitting} for a short review of the FRFT technique and the choice of FRFT parameters.\\
 The results of the GTS Parameter Estimation are summarised in Table \ref{tab:2}. It appears that the scale parameter $\beta$, $\alpha$, and $\lambda$ are all positive as expected; except $\beta_{-}$ for S\&P 500 return data, which is negative and suggests that the left side return variable ($X_{-}$) can be modelled by a compound Poisson distribution.
 
 \begin{table}[ht]
\vspace{-0.3cm}
\caption{Maximum Likelihood GTS Parameter Estimations }
\vspace{-0.3cm}
 \label{tab:2}
\centering
\begin{tabular}{@{}l l c c c @{}}
\toprule
\multirow{2}{*}{} &
\multirow{2}{*}{}&
\multicolumn{2}{c}{\textbf{Parameter Estimations per Index}}\\
\multirow{2}{*}{\textbf{Model}} &
\multirow{2}{*}{\textbf{Parameters}}&
\multirow{2}{*}{\textbf{500 SPY ETF return}}&
\multirow{2}{*}{\textbf{SPY ETF return}}& 
\multirow{2}{*}{\textbf{Bitcoin BTC return}}\\ 
\multirow{2}{*}{} &
\multirow{2}{*}{} &
\multirow{2}{*}{} &
\multirow{2}{*}{} &  
\multirow{2}{*}{}  \\ \toprule
\multirow{2}{*}{\textbf{GTS}} &
  \multirow{2}{*}{\textbf{$\mu$}} &
  \multirow{2}{*}{-0.5274011} &
  \multirow{2}{*}{-0.4145983} &
  \multirow{2}{*}{0.0284876*} \\
 \multirow{2}{*}{} &
  \multirow{2}{*}{\textbf{$\beta_{+}$}} &
  \multirow{2}{*}{0.5174702} &
  \multirow{2}{*}{0.5235145}&
  \multirow{2}{*}{-0.2560435}\\
  
  \multirow{2}{*}{} &
  \multirow{2}{*}{\textbf{$\beta_{-}$}} &
  \multirow{2}{*}{-0.0888191 } &
  \multirow{2}{*}{0.1531474} &
  \multirow{2}{*}{0.3863913}\\
  
  \multirow{2}{*}{} &
  \multirow{2}{*}{\textbf{$\alpha_{+}$}} &
  \multirow{2}{*}{0.6735391} &
  \multirow{2}{*}{0.6365290}&
   \multirow{2}{*}{1.2868131}\\
  
  \multirow{2}{*}{} &
  \multirow{2}{*}{\textbf{$\alpha_{-}$}} &
  \multirow{2}{*}{0.6083026} &
  \multirow{2}{*}{0.5118005}&
  \multirow{2}{*}{0.2771887}\\

  \multirow{2}{*}{} &
  \multirow{2}{*}{\textbf{$\lambda_{+}$}}&
  \multirow{2}{*}{1.2665066} &
  \multirow{2}{*}{1.2407793} &
  \multirow{2}{*}{3.7929526} \\
  
  \multirow{2}{*}{} &
  \multirow{2}{*}{\textbf{$\lambda_{-}$}}&
  \multirow{2}{*}{1.0807322} &
  \multirow{2}{*}{0.9354772} &
  \multirow{2}{*}{1.9676313} \\   
     
  \multirow{2}{*}{} &
  \multirow{2}{*}{} &
  \multirow{2}{*}{} &
  \multirow{2}{*}{} & 
  \multirow{2}{*}{} \\ 

\multirow{2}{*}{\textbf{GBM}} &
\multirow{2}{*}{\textbf{$\mu$}} &
\multirow{2}{*}{0.0501531} &
\multirow{2}{*}{0.0546935} &
\multirow{2}{*}{0.1790993} \\
    
  \multirow{2}{*}{} &
  \multirow{2}{*}{\textbf{$\sigma$}} &
  \multirow{2}{*}{0.9729045}& 
  \multirow{2}{*}{0.9669350} &
  \multirow{2}{*}{4.0287547} \\
   
\multirow{2}{*}{} & 
\multirow{2}{*}{} &
\multirow{2}{*}{} &
\multirow{2}{*}{} &
\multirow{2}{*}{} \\\bottomrule
 \end{tabular}
 \vspace{-0.2cm}
\end{table}

\noindent 
More details regarding the Parameter estimations are provided in Appendix \ref{eq:an1} , Table \ref{tab:5} for S\&P 500 return data. The maximization procedure convergences after 26 iterations; The last two columns of the  Table \ref{tab:5} show that the relationships (\ref{eq:l38}) are met. We have similar results in Appendix \ref{eq:an1}, Table \ref{tab:5} for SPY ETF and Appendix \ref{eq:an2}, Table \ref{tab:6} for Bitcoin BTC returns.\\
\noindent
For Bitcoin BTC, the location parameter $\mu*$ in Table \ref{tab:2} was replaced by 0.0665537 in order to align the theoretical and empirical Mean. The estimated parameters in Table \ref{tab:2} was used to compute the cumulants in Table \ref{tab:3}.
\begin{table}[ht]
\vspace{-0.3cm}
 \caption{Summary  Statistics }
\label{tab:3}
\centering
\begin{tabular}{@{}l c  c c c c@{}}
\toprule
 & \multicolumn{3}{c}{GTS(\textbf{$\beta_{+}$},\textbf{$\beta_{-}$},\textbf{$\alpha_{+}$},\textbf{$\alpha_{-}$},\textbf{$\lambda_{+}$},\textbf{$\lambda_{-}$})  Statistics} \\\toprule
Label & S\&P 500 return & SPY ETF  return & Bitcoin return\\ \toprule
Mean ($\kappa_{1}$)  & 0.0414355 & 0.0489189 & 0.1790993 \\
Variance ($\kappa_{2}$) & 0.9586858 & 0.9568384 & 15.5529035\\
Skewness ($\frac{\kappa_{3}}{\kappa_{2}^{3/2}}$) & -0.5843860 & -0.6322466 & -0.4119871 \\
Kurtosis ($3+\frac{\kappa_{4}}{\kappa_{2}^{2}}$) & 7.2851465 & 7.6521787 & 8.2745146 \\ \bottomrule
\end{tabular}%
 \vspace{-0.2cm}
\end{table}

\noindent 
When comparing empirical statistics from the raw data and GTS distribution statistics in Table \ref{tab:1} \& \ref{tab:3}, It appears that The estimation of the Maximum likelihood over-estimates the Variance, Kurtosis and skewness statistics.\\

\section{Kolmogorov-Smirnov (KS) Analysis}
\noindent
Given the sample of daily return $\{y_{1}, y_{2}\dots y_{m}\}$ of size m and the empirical cumulative distribution function (cdf) $F_{m}(x)$ for each index, the Kolmogorov-Smirnov (KS) test is performed under the null hypothesis ($H_{0}$) that the sample $\{y_{1}, y_{2}\dots y_{m}\}$ comes from the Generalized Tempered Sable  (GTS) distribution ($F(x)$). The cumulative distribution function of the theoretical distribution ($F(x)$) needs to be computed.  The density function ($f(x)$) does not have a closed-form, the same for the cumulative function($F(x)$) in (\ref{eq:l52}). However, we know the closed form of the Fourier of the density function ($\scrF[f]$) and the relationship in (\ref{eq:l53}) provides the Fourier of the cumulative distribution function ($\scrF[F]$). The GTS distribution function ($F(x)$) was computed from  the inverse of the Fourier of the cumulative distribution ($\scrF[F]$) in (\ref{eq:l54}). 
  \begin{align}
 Y &\sim GTS(\textbf{$\mu$}, \textbf{$\beta_{+}$}, \textbf{$\beta_{-}$}, \textbf{$\alpha_{+}$},\textbf{$\alpha_{-}$}, \textbf{$\lambda_{+}$}, f\textbf{$\lambda_{-}$}) \label{eq:l51}\\
   F(x)&= \int_{-\infty}^{x} f(t) \mathrm{d}t \hspace{5mm}  \hbox{$f$ is the density function of $Y$} \label{eq:l52}\\
 \scrF[F](x)&=\frac{\scrF[f](x)}{ix} + \pi\scrF[f](0)\delta (x) \label{eq:l53}\\
F(x) &= \frac{1} {2\pi}\int_{-\infty}^{+\infty}\! \frac{\scrF[f](y)}{iy}e^{ixy}\, \mathrm{d}y + \frac{1}{2}\label{eq:l54}
 \end{align}

 \noindent 
See Appendix A in \cite{nzokem2021fitting} for (\ref{eq:l53}) proof.\\
\noindent
The two-sided Kolmogorov-Smirnov goodness-of-fit statistic ($D_{m}$) is defined as follows.
 \begin{align}
D_{m} &= \sup_{x}{|F(x)-F_{m}(x)|}  &
P\_{value} &= prob(D_{m}>d_{m} |H_{0})  \label{eq:l55} 
 \end{align}
\noindent 
where $m$ is the sample size, $F_{m}(x)$ denotes the empirical cumulative distribution of $\{y_{1}, y_{2}\dots y_{m}\}$. The distribution of the Kolmogorov’s goodness-of-fit measure $D_{m}$ has been studied extensively in the litterature \cite{massey1951kolmogorov}. It was shown that the $D_{m}$ distribution is independent of the theoretical distribution ($F(x)$) under the null hypothesis (H0). \\
\begin{figure}[ht]
 \vspace{-0.4cm}
     \centering
         \includegraphics[scale=0.4] {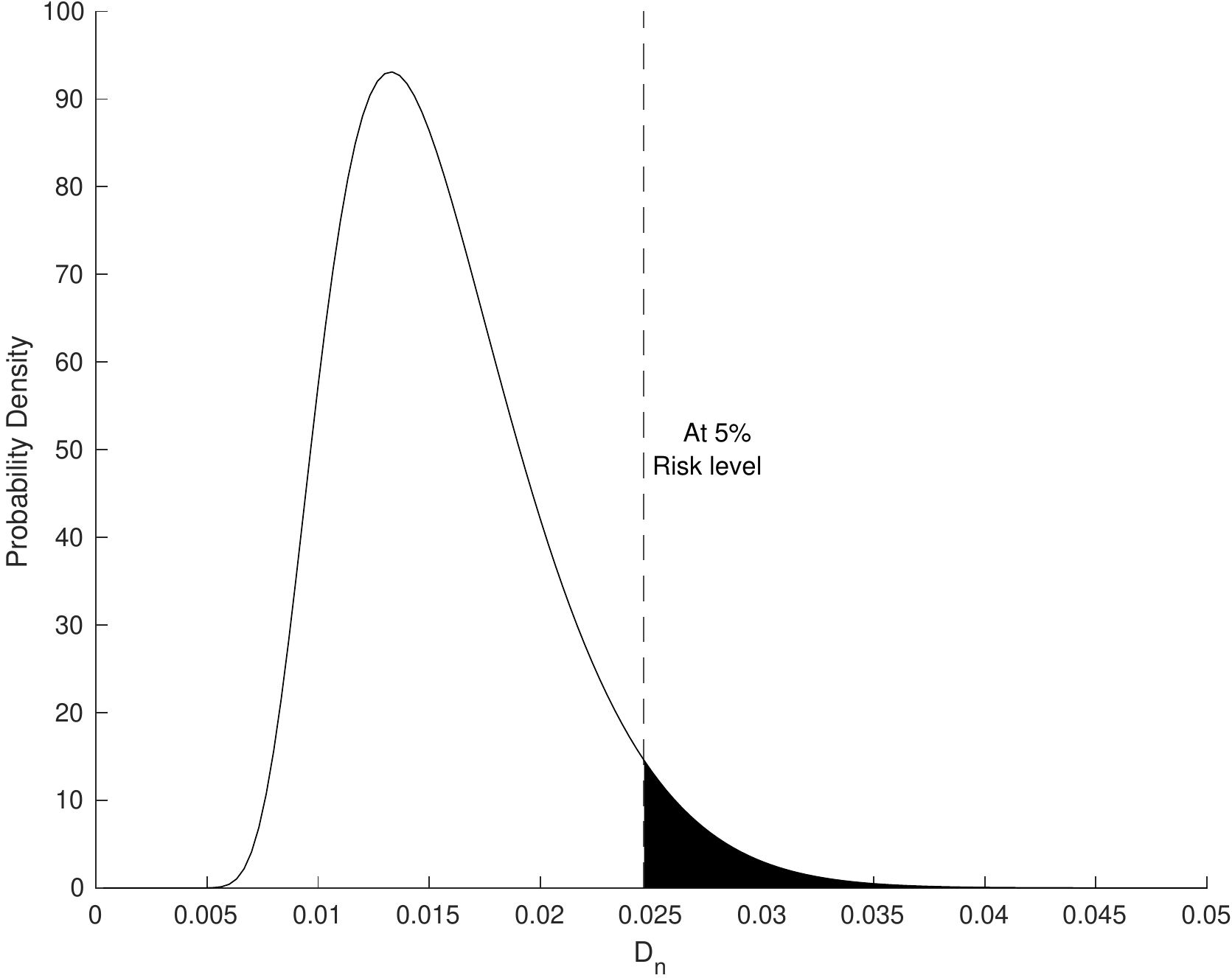}
 \vspace{-0.5cm}
        \caption{Kolmogorov-Smirnov Estimator ($\hat{D_{m}}$) probability density ($m=3048$)  under the null hypothesis $H_{0}$}
        \label{fig62}
 \vspace{-0.7cm}
\end{figure}

\noindent
 More recently, the numerical computation of the exact distribution of $D_{m}$ was developed  in \cite{dimitrova2020computing} along with the R package (KSgeneral). For $m=3048$, the probability density function  of $D_{m}$  was computed under the null hypothesis ($H_{0}$). As shown in Fig \ref{fig62}, the KS estimator ($D_{m}$) is a positively skewed distribution with a mean $\mu=0.0156$ and a standard deviation $\sigma=2.7*10^{-3}$. At $5\%$ risk level, the risk threshold is $d=0.0245$ and represented the area in the shaded area under the probability density function in Fig \ref{fig62}.\\
  \noindent
$d_{m}$ is the value of the KS estimator ($D_{m}$) computed from the sample $\{y_{1}, y_{2}\dots y_{m}\}$. Based on \cite{krysicki1999rachunek}, $d_{m}$  can be estimated as  follows.
\begin{align}
d_{m}&= Max(\sup_{0\leq j\leq P}{|F(x_{j})-F_{m}(x_{j})|}, \sup_{1\leq j\leq P}{|F(x_{j})-F_{m}(x_{j-1})|}) \label{eq:l56}
 \end{align}
For the SPY ETF return data, Appendix \ref{eq:an4} shows the detailed computation of the GTS cumulative distribution (F) and empirical cumulative distribution ($F_{n}$); and the quantity $d_{n}$ can be deduced.
\begin{align*}
\sup_{0\leq j\leq P}{|F(x_{j})-F_{m}(x_{j})|}=max(1)= 0.0141807 &\quad \sup_{1\leq j\leq P}{|F(x_{j})-F_{m}(x_{j-1})|}=max(2)=0.0225265	\\
d_{n}=0.0225265 &\quad P\_{value} = prob(D_{m}>d_{m} |H_{0})=8.95\%
\end{align*}
\noindent 
 For each Index and model, KS-Statistics ($d_{m}$) and P\_values were computed and summarized in Table \ref{tab:4}. See Appendix \ref{eq:an3}, Appendix \ref{eq:an4} and Appendix \ref{eq:an5} for details on the computation of KS-Statistics ($d_{m}$).
\begin{table}[ht]
\vspace{-0.3cm}
\centering
\caption{Kolmogorov-Smirnov (KS) test}
\vspace{-0.3cm}
 \label{tab:4}
\begin{tabular}{@{}lccl@{}}
\toprule
\textbf{Index} &
  \textbf{Model} &
  \textbf{$d_{m}$} &
  \textbf{P\_values} \\ \bottomrule
\textbf{S\&P500}    & \textbf{GTS}  & 0.0231677 & 7.48\% \\
                     & \textbf{GBM}  & 0.0911980 & 0\% \\
\textbf{SPY ETF} & \textbf{GTS} & 0.0225265 & 8.95\% \\
                   & \textbf{GBM}  & 0.0920545 & 0\% \\
\textbf{Bitcoin BTC} & \textbf{GTS} & 0.0264200 & 2.06\% \\
                   & \textbf{GBM}  & 0.1034062 & 0\% \\
\bottomrule
\end{tabular}
\vspace{-0.2cm}
\end{table}

\noindent
The standard Geometric Brownian motion (GBM) does not fit the S\&P 500, SPY ETF and Bitcoin BTC empirical distribution functions $F_{m}(x)$. The $P\_values$ are almost $ 0\%$, and the null hypothesis $H_{0}$ can not be accepted. The Generalized Tempered Stable distribution fits empirical distribution $F_{m}(x)$ at a different level of risk. As shown in Table \ref{tab:4}. In fact, the $P\_value$ is respectively $7.48\%$ for the $S\&P 500$ and  $8.95\%$ for the SPY ETF. These $P\_value$ are higher than the threshold $5\%$, and we can not reject the null hypothesis that the Generalized Tempered Stable distribution fits the empirical distribution. At $2\%$ risk level, the GTS distribution fits the Bitcoin BTC empirical distribution as well.\\
\noindent
The histogram of the daily return for each asset was compared graphically to the GTS and GBM density probability. As shown in Fig \ref{fig80}, GBM in red performs poorly. \\

\begin{figure}[ht]
\vspace{-0.5cm}
    \centering
  \begin{subfigure}[b]{0.31\linewidth}
    \includegraphics[width=\linewidth]{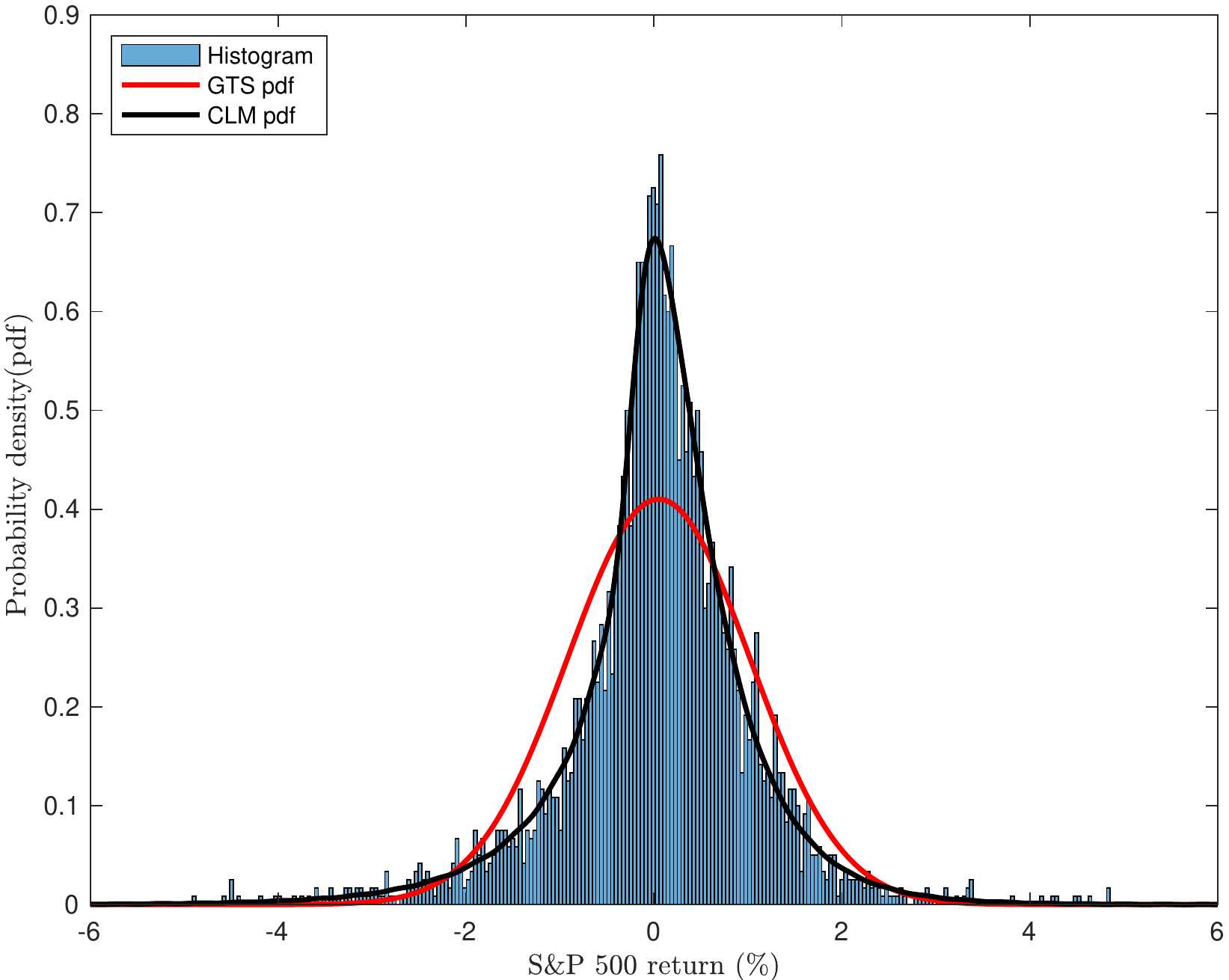}
\vspace{-0.5cm}
     \caption{Daily SP500 return (in \%)}
         \label{fig80}
  \end{subfigure}
  \begin{subfigure}[b]{0.31\linewidth}
    \includegraphics[width=\linewidth]{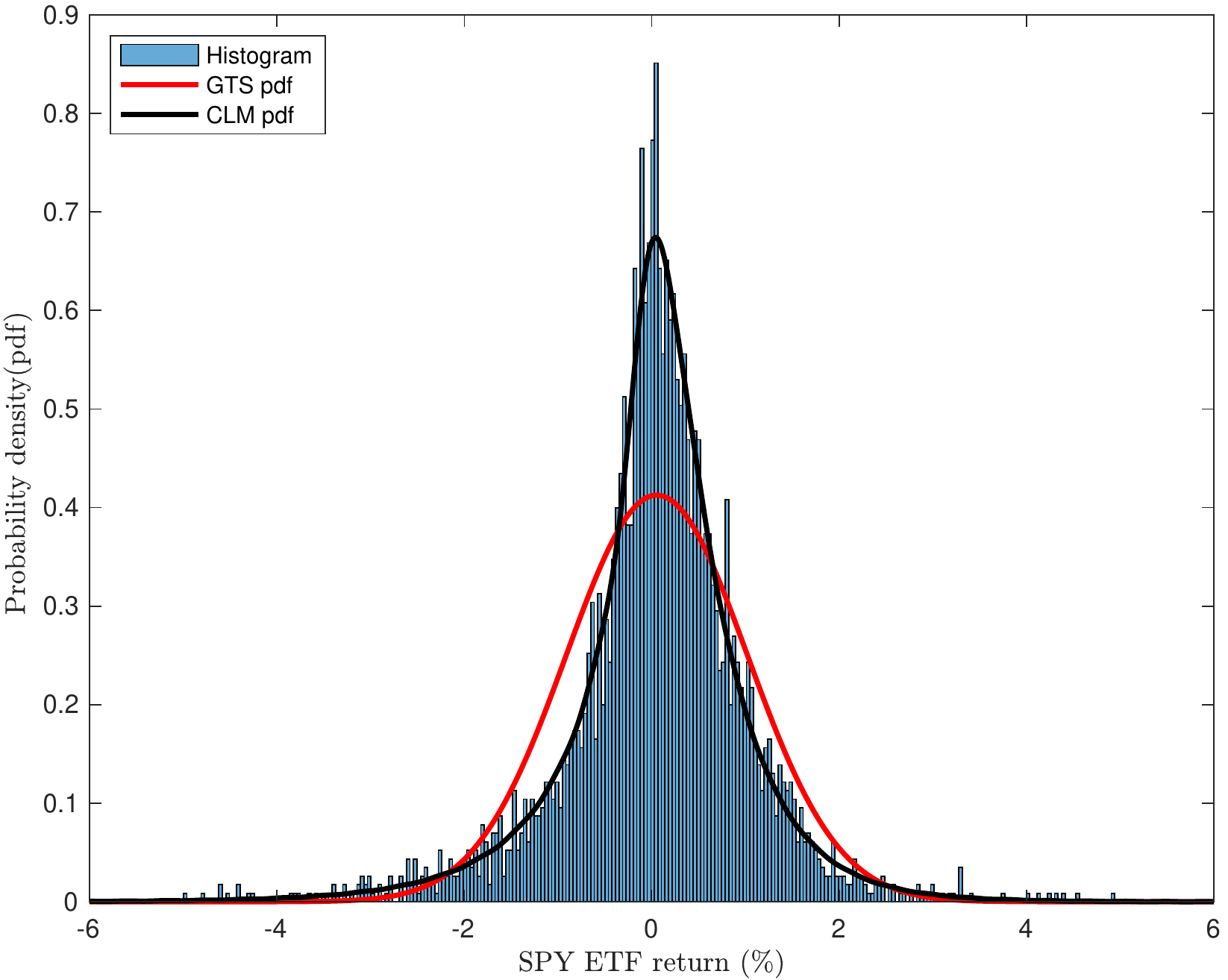}
\vspace{-0.5cm}
     \caption{ Daily SPY ETF return (in \%)}
         \label{fig81}
          \end{subfigure}
\begin{subfigure}[b]{0.31\linewidth}
    \includegraphics[width=\linewidth]{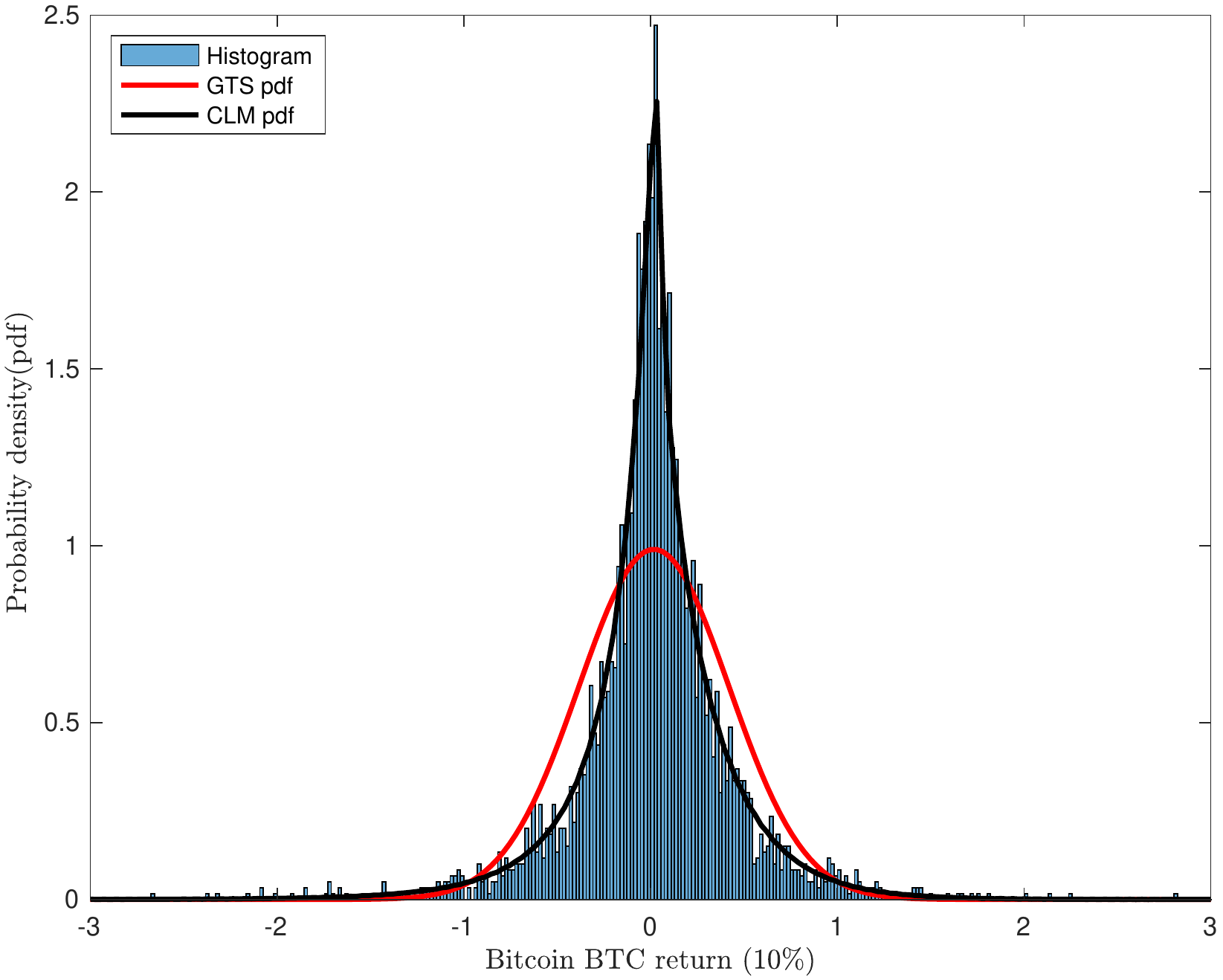}
\vspace{-0.5cm}
     \caption{ Daily Bitcoin return (in 10\%) }
         \label{fig82}
          \end{subfigure}
\vspace{-0.8cm}
  \caption{Daily Return}
  \label{fig80}
\vspace{-0.8cm}
\end{figure}
\section {Conclusion} 
\noindent 
The paper investigates the rich class of GTS distribution, an alternative to Normal distribution and the $\alpha$-Stable distribution for modelling asset return and many physical and economic systems.  We determine the expression of the characteristic exponent of the GTS distribution. We show that the Bilateral Gamma distribution and the Variance-Gamma (VG) distribution are special cases of the GTS distribution when $(\beta_{+},\beta_{-})$ goes to (0,0). Using the characteristic function, we show that the L\'evy process driven by a GTS distribution converges asymptotically to a L\'evy process driven by a normal distribution.
The empirical analysis shows that the GTS distribution fits the underlying distribution of the SPY ETF return. The right side of the Bitcoin BTC return and the left side of the S\&P 500 return underlying distributions fit the Tempered Stable distribution. The compound Poisson process models the left side of the Bitcoin BTC return and the right side of the S\&P 500 return underlying distributions. The Kolmogorov-Smirnov (KS) goodness-of-fit test shows that the GTS distribution fits the empirical distribution of the sample data at different risk levels. 
\section*{Acknowledgement}
\noindent 
We would like to thank an anonymous referee for valuable comments and suggestions which lead to the improvement of this version.
 \newpage
\pagestyle{plain}
\addcontentsline{toc}{chapter}{References}
\bibliographystyle{unsrt}
\bibliography{fmodelArxiv.bib}
\cleardoublepage
\appendix
\renewcommand{\thesection}{\Alph{section}}
\setcounter{section}{0}
\section{GTS(\textbf{$\beta_{+}$},\textbf{$\beta_{-}$},\textbf{$\alpha_{+}$},\textbf{$\alpha_{-}$},\textbf{$\lambda_{+}$},\textbf{$\lambda_{-}$}) Parameters Estimations \\
by Newton – Raphson Iteration Algorithm (\ref{eq:l39})}\label{eq:an1}

\begin{table}[ht]
\vspace{-0.3cm}
\caption{GTS(\textbf{$\beta_{+}$},\textbf{$\beta_{-}$},\textbf{$\alpha_{+}$},\textbf{$\alpha_{-}$},\textbf{$\lambda_{+}$},\textbf{$\lambda_{-}$}) Parameters Estimations for S\&P 500 data}
 \label{tab:5}
 \vspace{-0.3cm}
\centering
\resizebox{\textwidth}{!}{%
\begin{tabular}{|c|c|c|c|c|c|c|c|c|c|c|}
\hline
\textbf{Iterations} & \textbf{$\mu$} & \textbf{$\beta_{+}$} & \textbf{$\beta_{-}$} & \textbf{$\alpha_{+}$} & \textbf{$\alpha_{-}$} & \textbf{$\lambda_{+}$} & \textbf{$\lambda_{-}$} & \textbf{$Log(ML)$} & \textbf{$||\frac{dLog(ML)}{dV}||$} & \textbf{$Max Eigen Value$} \\ \hline
1 & 0 & 0.5 & 0.5 & 0.5 & 0.5 & 1 & 1 & -3971.8672 & 1706.60365 & 5647.82473 \\ \hline
2 & -0.0108585 & 0.49503397 & 0.4724221 & 0.47716693 & 0.38197317 & 1.00168003 & 0.68979404 & -3946.3049 & 175.738369 & -81.56681 \\ \hline
3 & -0.0392162 & 0.47801554 & 0.47596585 & 0.5472256 & 0.40902069 & 1.09483251 & 0.75236288 & -3942.9159 & 99.4792028 & 83.588565 \\ \hline
4 & -0.0481343 & 0.48217364 & 0.49316695 & 0.56805833 & 0.40535734 & 1.11725648 & 0.74911865 & -3942.6834 & 14.5626224 & -4.915284 \\ \hline
5 & -0.3111668 & 0.53220067 & 0.25437508 & 0.54133299 & 0.52174772 & 1.06556933 & 0.947431 & -3949.7853 & 669.849836 & 3470.16212 \\ \hline
6 & -0.3722117 & 0.5353265 & 0.2413161 & 0.5789964 & 0.46109859 & 1.11003293 & 0.85960717 & -3943.5893 & 691.166189 & 3336.50072 \\ \hline
7 & -0.3757992 & 0.53746917 & 0.27258653 & 0.59400245 & 0.45561584 & 1.15904589 & 0.85480648 & -3942.282 & 311.024149 & 1334.31509 \\ \hline
8 & -0.3720457 & 0.54080089 & 0.34703531 & 0.59877995 & 0.39874986 & 1.19012745 & 0.77333778 & -3941.9678 & 149.436218 & 456.462452 \\ \hline
9 & -0.4444434 & 0.54810179 & 0.25878792 & 0.6050535 & 0.42847021 & 1.19922515 & 0.82450996 & -3940.7561 & 60.2708444 & 28.6684195 \\ \hline
10 & -0.3441381 & 0.52272263 & 0.43004067 & 0.62442718 & 0.34439011 & 1.23644524 & 0.68635525 & -3943.024 & 124.344571 & -24.244076 \\ \hline
11 & -0.5017613 & 0.51996381 & 0.10696263 & 0.6398709 & 0.49260319 & 1.22989815 & 0.96109059 & -3944.2023 & 750.354439 & 79.1624455 \\ \hline
12 & -0.464922 & 0.52029261 & 0.16112398 & 0.64757436 & 0.45341583 & 1.244032 & 0.87562873 & -3939.8235 & 28.3740415 & -45.843028 \\ \hline
13 & -0.4517403 & 0.50625129 & 0.10735806 & 0.66489045 & 0.50000558 & 1.26048531 & 0.93820178 & -3939.2922 & 32.2342218 & 24.6252342 \\ \hline
14 & -0.5002366 & 0.52533523 & 0.03249099 & 0.65180095 & 0.53878505 & 1.24546973 & 0.99281121 & -3939.0864 & 15.734157 & 5.37669567 \\ \hline
15 & -0.554305 & 0.52957008 & -0.1212493 & 0.66309258 & 0.62522975 & 1.25533267 & 1.11114322 & -3938.9553 & 20.0970341 & -17.832874 \\ \hline
16 & -0.5456362 & 0.52724881 & -0.1113246 & 0.66450083 & 0.62023329 & 1.25725928 & 1.09604543 & -3938.9055 & 0.55668317 & -5.7529288 \\ \hline
17 & -0.5290624 & 0.51842697 & -0.0911012 & 0.67267289 & 0.60948744 & 1.26568886 & 1.08236135 & -3938.9017 & 3.34852458 & -4.4085232 \\ \hline
18 & -0.526767 & 0.51697167 & -0.0887426 & 0.67410608 & 0.60833828 & 1.26707679 & 1.08074557 & -3938.9016 & 0.05964325 & -6.1240703 \\ \hline
19 & -0.5274084 & 0.51747602 & -0.0888218 & 0.67353239 & 0.608303 & 1.26650004 & 1.0807331 & -3938.9016 & 0.00739881 & -6.1603863 \\ \hline
20 & -0.5274013 & 0.51747026 & -0.0888194 & 0.67353905 & 0.60830278 & 1.26650655 & 1.0807325 & -3938.9016 & 4.56E-06 & -6.1618571 \\ \hline
21 & -0.5274012 & 0.5174702 & -0.0888191 & 0.6735391 & 0.60830263 & 1.26650659 & 1.08073231 & -3938.9016 & 5.66E-07 & -6.1618686 \\ \hline
22 & -0.5274011 & 0.51747019 & -0.0888191 & 0.6735391 & 0.6083026 & 1.2665066 & 1.08073228 & -3938.9016 & 9.25E-08 & -6.1618702 \\ \hline
23 & -0.5274011 & 0.51747019 & -0.0888191 & 0.67353911 & 0.6083026 & 1.2665066 & 1.08073227 & -3938.9016 & 1.50E-08 & -6.1618705 \\ \hline
24 & -0.5274011 & 0.51747019 & -0.0888191 & 0.67353911 & 0.6083026 & 1.2665066 & 1.08073227 & -3938.9016 & 2.45E-09 & -6.1618705 \\ \hline
25 & -0.5274011 & 0.51747019 & -0.0888191 & 0.67353911 & 0.6083026 & 1.2665066 & 1.08073227 & -3938.9016 & 3.87E-10 & -6.1618705 \\ \hline
26 & -0.5274011 & 0.51747019 & -0.0888191 & 0.67353911 & 0.6083026 & 1.2665066 & 1.08073227 & -3938.9016 & 6.67E-11 & -6.1618705 \\ \hline
\end{tabular}%
}
\end{table}

\begin{table}[ht]
\vspace{-0.3cm}
\centering
\caption{GTS(\textbf{$\beta_{+}$},\textbf{$\beta_{-}$},\textbf{$\alpha_{+}$},\textbf{$\alpha_{-}$},\textbf{$\lambda_{+}$},\textbf{$\lambda_{-}$}) Parameters Estimations for SPY ETF data}
\label{tab:6}
\vspace{-0.3cm}
\resizebox{\textwidth}{!}{%
\begin{tabular}{|c|c|c|c|c|c|c|c|c|c|c|}
\hline
\textbf{$Iterations$} & \textbf{$\mu$} & \textbf{$\beta_{+}$} & \textbf{$\beta_{-}$} & \textbf{$\alpha_{+}$} & \textbf{$\alpha_{-}$} & \textbf{$\lambda_{+}$} & \textbf{$\lambda_{-}$} & \textbf{$Log(ML)$} & \textbf{$||\frac{dLog(ML)}{dV}||$} & \textbf{$Max Eigen Value$} \\ \hline
1 & 0.03117789 & 0.38858855 & 0.45564901 & 0.63536034 & 0.44065667 & 1.19999352 & 0.80066172 & -3925.3636 & 71.5597807 & -0.271327 \\ \hline
2 & 0.00284966 & 0.3984325 & 0.44341049 & 0.63580752 & 0.44151872 & 1.20308769 & 0.80372187 & -3924.9767 & 64.0358701 & -1.3348134 \\ \hline
3 & -0.0188742 & 0.40608802 & 0.43162621 & 0.63492513 & 0.44433412 & 1.20357975 & 0.80928985 & -3924.7491 & 60.1230479 & -1.9189062 \\ \hline
4 & -0.0428858 & 0.41425675 & 0.41901201 & 0.63456553 & 0.44678346 & 1.20502555 & 0.81447863 & -3924.4959 & 55.6344633 & -2.4628347 \\ \hline
5 & -0.0673489 & 0.42245611 & 0.40593159 & 0.63430173 & 0.44925856 & 1.20667157 & 0.81975434 & -3924.2538 & 51.2562424 & -2.9006588 \\ \hline
6 & -0.0920084 & 0.43064823 & 0.39237468 & 0.6340158 & 0.45187601 & 1.20833783 & 0.82527098 & -3924.0284 & 47.0663973 & -3.2476602 \\ \hline
7 & -0.1169168 & 0.43885448 & 0.37825282 & 0.63368414 & 0.45468853 & 1.20999872 & 0.83110781 & -3923.8194 & 43.0341914 & -3.5267786 \\ \hline
8 & -0.1422131 & 0.44710211 & 0.36343322 & 0.6333181 & 0.45774079 & 1.21168158 & 0.83733848 & -3923.6251 & 39.1128746 & -3.757145 \\ \hline
9 & -0.1680813 & 0.45542017 & 0.34774111 & 0.63294478 & 0.46108339 & 1.2134386 & 0.84404886 & -3923.4443 & 35.2548844 & -3.9528906 \\ \hline
10 & -0.1947404 & 0.46384084 & 0.33095347 & 0.63259883 & 0.46478049 & 1.21533585 & 0.85134731 & -3923.2757 & 31.4143849 & -4.1238109 \\ \hline
11 & -0.2224442 & 0.47240032 & 0.31278605 & 0.6323193 & 0.46891728 & 1.21744973 & 0.8593752 & -3923.119 & 27.5471997 & -4.2762154 \\ \hline
12 & -0.2514828 & 0.48113769 & 0.29287359 & 0.63214977 & 0.47360991 & 1.21986745 & 0.86832067 & -3922.9746 & 23.6109225 & -4.4134458 \\ \hline
13 & -0.2821707 & 0.49008743 & 0.2707478 & 0.63214123 & 0.4790185 & 1.2226895 & 0.87843537 & -3922.8441 & 19.5680604 & -4.5357431 \\ \hline
14 & -0.3147707 & 0.49925064 & 0.24584352 & 0.63235833 & 0.4853581 & 1.22602959 & 0.89004186 & -3922.7309 & 15.4004254 & -4.6389257 \\ \hline
15 & -0.3491421 & 0.50848776 & 0.21769799 & 0.6328892 & 0.49286573 & 1.22999006 & 0.90345435 & -3922.6411 & 11.1634341 & -4.711419 \\ \hline
16 & -0.3832555 & 0.51712008 & 0.18716565 & 0.63384475 & 0.50147935 & 1.2345029 & 0.91839631 & -3922.5839 & 7.14673937 & -4.7366424 \\ \hline
17 & -0.4086577 & 0.5228833 & 0.16128441 & 0.6352204 & 0.50929013 & 1.23867111 & 0.9314548 & -3922.564 & 3.6912078 & -4.7507677 \\ \hline
18 & -0.414826 & 0.52369643 & 0.15305883 & 0.63632818 & 0.51190794 & 1.24055798 & 0.93563502 & -3922.5626 & 0.52220937 & -4.824878 \\ \hline
19 & -0.4145758 & 0.52350521 & 0.15317187 & 0.63653526 & 0.51178869 & 1.24078564 & 0.9354597 & -3922.5626 & 0.00088604 & -4.8353435 \\ \hline
20 & -0.4146011 & 0.52351555 & 0.15314403 & 0.6365283 & 0.5118022 & 1.24077848 & 0.93547969 & -3922.5626 & 4.86E-05 & -4.8354104 \\ \hline
21 & -0.4145979 & 0.52351432 & 0.15314792 & 0.63652914 & 0.51180025 & 1.24077938 & 0.9354768 & -3922.5626 & 6.74E-06 & -4.835384 \\ \hline
22 & -0.4145984 & 0.5235145 & 0.15314736 & 0.63652902 & 0.51180052 & 1.24077925 & 0.93547721 & -3922.5626 & 9.58E-07 & -4.8353877 \\ \hline
23 & -0.4145983 & 0.52351447 & 0.15314744 & 0.63652904 & 0.51180048 & 1.24077927 & 0.93547715 & -3922.5626 & 1.36E-07 & -4.8353871 \\ \hline
24 & -0.4145983 & 0.52351447 & 0.15314743 & 0.63652903 & 0.51180049 & 1.24077927 & 0.93547716 & -3922.5626 & 1.93E-08 & -4.8353872 \\ \hline
25 & -0.4145983 & 0.52351447 & 0.15314743 & 0.63652903 & 0.51180049 & 1.24077927 & 0.93547716 & -3922.5626 & 2.74E-09 & -4.8353872 \\ \hline
\end{tabular}%
}
\end{table}
\newpage
\section{GTS(\textbf{$\beta_{+}$},\textbf{$\beta_{-}$},\textbf{$\alpha_{+}$},\textbf{$\alpha_{-}$},\textbf{$\lambda_{+}$},\textbf{$\lambda_{-}$}) Parameters Estimations \\
by Newton – Raphson Iteration Algorithm (\ref{eq:l39})}\label{eq:an2}

\begin{table}[ht]
\vspace{-0.3cm}
\centering
\caption{GTS(\textbf{$\beta_{+}$},\textbf{$\beta_{-}$},\textbf{$\alpha_{+}$},\textbf{$\alpha_{-}$},\textbf{$\lambda_{+}$},\textbf{$\lambda_{-}$}) Parameters Estimations for Bitcoin BTC}
\label{tab:7}
\vspace{-0.3cm}
\resizebox{\textwidth}{!}{%
\begin{tabular}{|c|c|c|c|c|c|c|c|c|c|c|}
\hline
\textbf{$Iterations$} & \textbf{$\mu$} & \textbf{$\beta_{+}$} & \textbf{$\beta_{-}$} & \textbf{$\alpha_{+}$} & \textbf{$\alpha_{-}$} & \textbf{$\lambda_{+}$} & \textbf{$\lambda_{-}$} & \textbf{$Log(ML)$} & \textbf{$||\frac{dLog(ML)}{dV}||$} & \textbf{$Max Eigen Value$} \\ \hline
1  & 0.21185356 & -0.5936847 & 0.68060815 & 2.13344539 & 0.17485904 & 4.32852817 & 1.6632197  & -1113.2925 & 391.964759 & 109.54199  \\ \hline
2  & 0.20179463 & -0.5719942 & 0.66869888 & 2.06782696 & 0.17926313 & 4.25935989 & 1.68061739 & -1112.948  & 66.3617005 & 39.7235863 \\ \hline
3  & 0.18488156 & -0.4805755 & 0.64904251 & 1.67435697 & 0.18725647 & 4.04431209 & 1.7243988  & -1113.298  & 375.115614 & 47.2580752 \\ \hline
4  & 0.17359957 & -0.4599726 & 0.62737902 & 1.64512058 & 0.19924724 & 4.00532655 & 1.80054485 & -1113.2767 & 34.0034116 & 13.7907312 \\ \hline
5  & 0.15654461 & -0.3631565 & 0.60466027 & 1.33550106 & 0.21009227 & 3.80355108 & 1.86648876 & -1113.8264 & 297.165624 & 23.4660262 \\ \hline
6  & 0.13118124 & -0.3416734 & 0.54091103 & 1.31066246 & 0.24980581 & 3.76779634 & 2.10084302 & -1115.0936 & 24.4310009 & 49.6898202 \\ \hline
7  & 0.14486055 & -0.315749  & 0.59348696 & 1.25030678 & 0.2126025  & 3.71622561 & 1.91347747 & -1113.7821 & 180.934774 & 108.783891 \\
\hline
\end{tabular}%
}
\end{table}

\newpage
\section{S\&P 500 return data\\GTS Cumulative distribution ($F$)  versus Empirical distribution ($\hat{F_{n}}$)\\ $(1)=|\hat{F(x_{j})}-\hat{F_{n}(x_{j-1})}|$,$(2)=|\hat{F(x_{j})}-\hat{F_{n}(x_{j})}|$, $n=3046$}\label{eq:an3}
 \label{tab:9}
\begin{table}[ht]
\centering
\resizebox{14cm}{!}{%
\begin{tabular}{llllllllllllllllll}
\textbf{$x_{j}$} &
  \textbf{$n\_\{j\}$} &
  \textbf{$\hat{F_{n}(x_{j})}$} &
  \textbf{$\hat{F(x_{j})}$} &
  \textbf{(1)} &
  \textbf{(2)} &
 \textbf{$x_{j}$} &
  \textbf{$n\_\{j\}$} &
  \textbf{$\hat{F_{n}(x_{j})}$} &
  \textbf{$\hat{F(x_{j})}$} &
   \textbf{(1)} &
  \textbf{(2)} &
  \textbf{$x_{j}$} &
  \textbf{$n\_\{j\}$} &
  \textbf{$\hat{F_{n}(x_{j})}$} &
  \textbf{$\hat{F(x_{j})}$} &
  \textbf{(1)} &
  \textbf{(2)} \\
 &
   &
   &
   &
   &
   &
   &
   &
   &
   &
   &
   &
   &
   &
   &
   &
   &
   \\
-5.15 &
  0 &
  0 &
  0.00096927 &
  0.00096927 &
  0.00096927 &
  -1.7222 &
  6 &
  0.04661852 &
  0.04466376 &
  1.50E-05 &
  0.00195476 &
  1.7056 &
  6 &
  0.96946815 &
  0.96329883 &
  0.00419953 &
  0.00616933 \\
-5.1106 &
  0 &
  0 &
  0.0010153 &
  0.0010153 &
  0.0010153 &
  -1.6828 &
  9 &
  0.04957321 &
  0.04667203 &
  5.35E-05 &
  0.00290118 &
  1.745 &
  7 &
  0.97176625 &
  0.96537441 &
  0.00409375 &
  0.00639184 \\
-5.0712 &
  0 &
  0 &
  0.00105086 &
  0.00105086 &
  0.00105086 &
  -1.6434 &
  9 &
  0.05252791 &
  0.04881667 &
  0.00075655 &
  0.00371124 &
  1.7844 &
  4 &
  0.97307945 &
  0.9673409 &
  0.00442535 &
  0.00573855 \\
-5.0318 &
  0 &
  0 &
  0.00107641 &
  0.00107641 &
  0.00107641 &
  -1.604 &
  9 &
  0.0554826 &
  0.0511055 &
  0.00142241 &
  0.0043771 &
  1.8238 &
  3 &
  0.97406435 &
  0.96921271 &
  0.00386674 &
  0.00485164 \\
-4.9924 &
  0 &
  0 &
  0.00109407 &
  0.00109407 &
  0.00109407 &
  -1.5646 &
  7 &
  0.0577807 &
  0.05354082 &
  0.00194178 &
  0.00423987 &
  1.8632 &
  6 &
  0.97603414 &
  0.97186382 &
  0.00220052 &
  0.00417032 \\
-4.953 &
  0 &
  0 &
  0.00110734 &
  0.00110734 &
  0.00110734 &
  -1.5252 &
  8 &
  0.06040709 &
  0.05611954 &
  0.00166115 &
  0.00428755 &
  1.9026 &
  6 &
  0.97800394 &
  0.97353351 &
  0.00250064 &
  0.00447043 \\
\end{tabular}%
}
\end{table}

\newpage
\section{SPY ETF return data\\GTS Cumulative distribution ($F$)  versus Empirical distribution ($\hat{F_{n}}$)\\ $(1)=|\hat{F(x_{j})}-\hat{F_{n}(x_{j-1})}|$,$(2)=|\hat{F(x_{j})}-\hat{F_{n}(x_{j})}|$$n=3046$}\label{eq:an4}
 \label{tab:9}
\begin{table}[ht]
\centering
\resizebox{14cm}{!}{%
\begin{tabular}{llllllllllllllllll}
\textbf{$x_{j}$} &
  \textbf{$n_{j}$} &
  \textbf{$\hat{F_{n}(x_{j})}$} &
  \textbf{$\hat{F(x_{j})}$} &
  \textbf{(1)} &
  \textbf{(2)} &
 \textbf{$x_{j}$} &
  \textbf{$n_{j}$} &
  \textbf{$\hat{F_{n}(x_{j})}$} &
  \textbf{$\hat{F(x_{j})}$} &
   \textbf{(1)} &
  \textbf{(2)} &
  \textbf{$x_{j}$} &
  \textbf{$n_{j}$} &
  \textbf{$\hat{F_{n}(x_{j})}$} &
  \textbf{$\hat{F(x_{j})}$} &
  \textbf{(1)} &
  \textbf{(2)} \\
 &
   &
   &
   &
   &
   &
   &
   &
   &
   &
   &
   &
   &
   &
   &
   &
   &
   \\
-5.15 &
  0 &
  0 &
  0.001058 &
  0.001058 &
  0.001058 &
  -1.748 &
  2 &
  0.04497702 &
  0.0427068 &
  0.00161362 &
  0.00227022 &
  1.654 &
  8 &
  0.96717006 &
  0.96017174 &
  0.00437192 &
  0.00699832 \\
-5.1122 &
  0 &
  0 &
  0.00111651 &
  0.00111651 &
  0.00111651 &
  -1.7102 &
  8 &
  0.04760341 &
  0.04457001 &
  0.00040701 &
  0.0030334 &
  1.6918 &
  7 &
  0.96946815 &
  0.96245719 &
  0.00471287 &
  0.00701097 \\
-5.0744 &
  0 &
  0 &
  0.00116784 &
  0.00116784 &
  0.00116784 &
  -1.6724 &
  8 &
  0.05022981 &
  0.04654588 &
  0.00105753 &
  0.00368392 &
  1.7296 &
  6 &
  0.97143795 &
  0.964592 &
  0.00487616 &
  0.00684596 \\
-5.0366 &
  0 &
  0 &
  0.00121066 &
  0.00121066 &
  0.00121066 &
  -1.6346 &
  10 &
  0.0535128 &
  0.04864851 &
  0.0015813 &
  0.0048643 &
  1.7674 &
  5 &
  0.97307945 &
  0.96659508 &
  0.00484287 &
  0.00648437 \\
-4.9988 &
  1 &
  0.0003283 &
  0.00124504 &
  0.00124504 &
  0.00091675 &
  -1.5968 &
  3 &
  0.0544977 &
  0.05088868 &
  0.00262412 &
  0.00360902 &
  1.8052 &
  4 &
  0.97439265 &
  0.96848431 &
  0.00459514 &
  0.00590833 \\
-4.961 &
  0 &
  0.0003283 &
  0.00127251 &
  0.00094421 &
  0.00094421 &
  -1.559 &
  6 &
  0.0564675 &
  0.05327299 &
  0.00122471 &
  0.00319451 &
  1.843 &
  3 &
  0.97537754 &
  0.97113776 &
  0.00325488 &
  0.00423978 \\
\end{tabular}%
}
\end{table}

\newpage

\section{Bitcoin BTC return data\\GTS Cumulative distribution ($F$)  versus Empirical distribution ($\hat{F_{n}}$)\\ $(1)=|\hat{F(x_{j})}-\hat{F_{n}(x_{j-1})}|$,$(2)=|\hat{F(x_{j})}-\hat{F_{n}(x_{j})}|$, $n=3267$}\label{eq:an5}
 \label{tab:9}
\begin{table}[ht]
\centering
\resizebox{14.25cm}{!}{%
\begin{tabular}{llllllllllllllllll}
\textbf{$x_{j}$} &
  \textbf{$n_{j}$} &
  \textbf{$\hat{F_{n}(x_{j})}$} &
  \textbf{$\hat{F(x_{j})}$} &
  \textbf{(1)} &
  \textbf{(2)} &
 \textbf{$x_{j}$} &
  \textbf{$n_{j}$} &
  \textbf{$\hat{F_{n}(x_{j})}$} &
  \textbf{$\hat{F(x_{j})}$} &
   \textbf{(1)} &
  \textbf{(2)} &
  \textbf{$x_{j}$} &
  \textbf{$n_{j}$} &
  \textbf{$\hat{F_{n}(x_{j})}$} &
  \textbf{$\hat{F(x_{j})}$} &
  \textbf{(1)} &
  \textbf{(2)} \\
 &
   &
   &
   &
   &
   &
   &
   &
   &
   &
   &
   &
   &
   &
   &
   &
   &
   \\
  -2.602&
  0 &
  0.00061218 &
  0.00014873 &
  0.00046346 &
  0.00046346 &
  -0.9822 &
  2 &
  0.02142639 &
  0.01842616 &
  0.00238804 &
  0.00300022 &
  0.6376 &
  14 &
  0.95102541 &
  0.94926001 &
  0.00251988 &
  0.00176539 \\
-2.5838 &
  0 &
  0.00061218 &
  0.0003145 &
  0.00029768 &
  0.00029768 &
  -0.964 &
  0 &
  0.02142639 &
  0.01965368 &
  0.0017727 &
  0.0017727 &
  0.6558 &
  6 &
  0.95286195 &
  0.95263798 &
  0.00161257 &
  0.00022397 \\
-2.5656 &
  0 &
  0.00061218 &
  0.00048303 &
  0.00012915 &
  0.00012915 &
  -0.9458 &
  2 &
  0.02203857 &
  0.02080114 &
  0.00062524 &
  0.00123742 &
  0.674 &
  11 &
  0.95622896 &
  0.95601614 &
  0.00315419 &
  0.00021282 \\
-2.5474 &
  0 &
  0.00061218 &
  0.00063146 &
  1.93E-05 &
  1.93E-05 &
  -0.9276 &
  6 &
  0.02387511 &
  0.0218493 &
  0.00018926 &
  0.00202581 &
  0.6922 &
  5 &
  0.95775941 &
  0.95934677 &
  0.00311781 &
  0.00158735 \\
-2.5292 &
  0 &
  0.00061218 &
  0.00073969 &
  0.00012751 &
  0.00012751 &
  -0.9094 &
  3 &
  0.02479339 &
  0.02280256 &
  0.00107256 &
  0.00199083 &
  0.7104 &
  9 &
  0.96051423 &
  0.96256256 &
  0.00480315 &
  0.00204832 \\
-2.511 &
  0 &
  0.00061218 &
  0.0007934 &
  0.00018122 &
  0.00018122 &
  -0.8912 &
  4 &
  0.02601775 &
  0.02368989 &
  0.0011035 &
  0.00232787 &
  0.7286 &
  9 &
  0.96326905 &
  0.96558926 &
  0.00507503 &
  0.00232021 \\
\end{tabular}%
}
\end{table}

\end{document}